\documentclass[10pt,a4paper]{article}

\usepackage[utf8]{inputenc}
\usepackage{blindtext, graphicx, subfigure, amsmath, wrapfig, siunitx, tabto, amssymb, cite, ragged2e, multicol, tikz, physics, color, tcolorbox, amsfonts, algorithm, algpseudocode, amsthm, hyperref}
\usetikzlibrary{shapes.geometric, arrows}

\newtheorem{theorem}{Theorem}

\newtheorem{proposition}{Proposition}
\newtheorem{assumption}{Assumption}

\theoremstyle{definition}
\newtheorem{definition}{Definition}
\theoremstyle{remark}
\newtheorem{remark}{Remark}

\usepackage[top=0.75in, bottom=1in, left=0.75in, right=0.75in]{geometry}

\tikzstyle{startstop} = [rectangle, rounded corners, minimum width=1cm, minimum height=1cm,text centered, draw=black, fill=red!30]
\tikzstyle{io} = [rectangle, minimum width=1.5cm, minimum height=1cm, text centered, draw=black, fill=green!10]
\tikzstyle{process} = [rectangle, minimum width=1.5cm, minimum height=1cm, text width=3.5cm, text centered, draw=black, fill=orange!30]
\tikzstyle{summer} = [circle, minimum width=0.75cm, minimum height=0.75cm, draw=black, fill=green!1]
\tikzstyle{route} = [circle, fill, inner sep=1.5pt]
\tikzstyle{arrow} = [thick,->,>=stealth]

\begin{document}
	
\begin{center}
	\Large
	\textbf{Robust Model Reference Adaptive Control with Combined Adaptation under Finite Excitation Condition}
\end{center}
	
\begin{center}
	Manish Patel and Arnab Maity \\
	\textit{
		Indian Institute of Technology Bombay, Mumbai, India
	}	
\end{center}

\section*{Abstract}
In adaptive control, parametric uncertainties in linear-in-parameter form consist of unknown parameters and known regressor signals. Convergence of the unknown parameters to their ideal values requires the regressor to satisfy a persistent excitation (PE) condition, which depends on future data and is therefore infeasible to guarantee online. Memory-based parameter update laws address this by enabling ideal parameter convergence under the online-verifiable finite excitation (FE) condition. In this paper, a new algorithm is proposed to construct a memory term via the Modified Gram-Schmidt orthogonalization procedure for a class of multi-input multi-output nonlinear systems with an unknown diagonal control effectiveness matrix and bounded nonparametric uncertainties. Under the finite excitation condition, the constructed memory term yields an identity coefficient matrix in the parameter estimation error dynamics. The identity coefficient matrix eliminates the need for time-varying adaptation gains, enables an explicit ultimate bound on the parameter estimation error, and preserves the structure of the nonparametric uncertainty bound under the memory term. Building on this, a combined adaptation law is developed for controller gain estimation under FE. The closed-loop tracking and estimation errors are shown to decay exponentially to a neighborhood of the origin, characterized by an explicit ultimate bound, with a decay rate that depends solely on user-defined gains and system constants, independent of the level of regressor excitation. This removes the dependence of the convergence rate on the level of regressor excitation, a key limitation of existing approaches such as concurrent learning, memory regressor extension, and DREM. 

\textbf{Keywords:} Combined Adaptive Control, Robust Model Reference Adaptive Control, Unknown Control Effectiveness Matrix, Finite Excitation, Exponential Parameter Convergence.
	
\section{Introduction}
\label{section:Introduction}

The design of adaptive controllers for solving regulation and tracking problems is well-established for systems with parametric uncertainties \cite{narendra2012stable, ioannou1996robust, sastry1989adaptive}. Such controllers estimate unknown parameters to match the system dynamics or update controller gains to achieve the desired objectives. In practical systems, the most commonly occurring parametric uncertainties can be expressed in a linear-in-parameter form, involving unknown parameters and known regressor signals. These unknown parameters are estimated online to achieve the primary objective of stable regulation and tracking. As long as the primary objective is fulfilled, convergence of the estimated parameters to their ideal values is not strictly necessary. Boundedness of the parameter estimates is sufficient to ensure closed-loop stability and asymptotic regulation or tracking. However, in the presence of nonparametric uncertainties, the parameter estimates may diverge, leading to unbounded behavior \cite{ioannou1996robust}. This divergence can destabilize the closed-loop system and prevent the controller from achieving the desired objectives. Consequently, without robustness modifications, conventional adaptive controllers may lack robustness to nonparametric uncertainties. \par

Robustness modification terms in adaptation laws are designed to prevent the parameter estimates from diverging \cite[Chapter 8]{lavretsky2013robust}. However, each of these modifications have trade-offs. Some require prior knowledge of upper bounds on unknown parameters or nonparametric uncertainties, while others hinder the ideal parameter convergence even when such uncertainties are absent. Convergence of the parameter estimates to their ideal values not only ensures robustness against nonparametric uncertainties but also enables unique system identification. The necessary and sufficient condition for parameter convergence is that the regressor signal satisfies the persistent excitation (PE) condition \cite{narendra1987persistent}. However, PE inherently requires knowledge of future system data, rendering it impractical to verify or ensure online. In tracking problems,  this stringent condition is translated into requirements on the reference signal. If the reference signal contains as many spectral lines as the number of unknown parameters, the regressor is said to be PE \cite{boyd1986necessary}. Recent work has proposed transient enrichment of the reference signal to relax the PE condition for robust adaptive controllers \cite{gallegos2024relaxed}. Nevertheless, reference signals are typically not designer-specified but are determined by the system's task or mission. While adding small perturbation signals may help ensure sufficient spectral content for parameter convergence, such modifications can degrade overall system performance. Therefore, achieving parameter convergence by imposing conditions on reference signals is generally not a practical solution in adaptive control applications. \par

Over the past decade, the adaptive control research community has made significant efforts to derive feasible conditions that are less restrictive than persistent excitation (PE) yet still guarantee parameter convergence. One notable advancement is the introduction of a memory term in the adaptation law, which updates parameter estimates using both instantaneous and past system data \cite{chowdhary2013concurrent}, \cite{hocht2015frequency}, \cite{roy2017uges}, \cite{cho2017composite}, \cite{pan2018composite}. These developments build upon earlier pioneering works such as \cite{lion1967rapid} and \cite{kreisselmeier1977adaptive}. The memory term enables the relaxation of the PE condition by replacing it with a finite excitation (FE) condition. The FE condition is milder than PE because it relies solely on past data and does not require future information, making it verifiable in real time. Gradient-based and least-squares-based parameter update laws, when augmented with memory terms, can guarantee exponential convergence of the parameter estimation error to the origin under the FE condition, hence, the robustness. These contributions represent substantial progress toward developing practical conditions for ensuring parameter convergence to their ideal values. However, two critical challenges remain. The first issue is that the convergence rate of the parameter estimation error depends directly on the minimum eigenvalue of the coefficient matrix in the parameter estimation error dynamics. Although time-varying adaptation gains have been proposed to accelerate convergence \cite{kim2022data}, \cite{patel2022parameter}, \cite{glushchenko2022exponentially}, they can become excessively large in the case of ill-conditioned matrices, leading to stiffness in the error dynamics \cite{ioannou1996robust}. While normalizing by the maximum eigenvalue can alleviate stiffness, it significantly slows the convergence rate. The second issue arises from the fact that the memory term transforms the nonparametric uncertainty vector using a memory matrix. Since this matrix depends on closed-loop signals, it becomes difficult to derive an explicit ultimate bound on the estimation error. As a result, although robustness to nonparametric uncertainty can be established, the size of the ultimate bound near the origin remains dependent on the closed-loop behavior. \par

In an effort to partially address the aforementioned issues, Dynamic Regressor Extension and Mixing (DREM) \cite{aranovskiy2016parameters} constructs an instantaneous square matrix by processing system data through several stable filters, which are subsequently mixed to supply a time-varying diagonal coefficient matrix to the parameter estimation error dynamics.  Under the finite excitation (FE) condition, the diagonal elements of this matrix are strictly negative \cite{yi2022conditions} and are equal to the square of the determinant of the instantaneous matrix. Consequently, the convergence rate of the parameter estimation error is proportional to the square of this determinant. For  ill-conditioned matrices, however, the determinant and thus the convergence rate can become arbitrarily small, necessitating large time-varying adaptation gains to achieve meaningful convergence \cite{glushchenko2022exponentially}. Using the inverse of a small determinant increases the sensitivity of the closed-loop system to the adaptation gain, which may compromise stability. Therefore, while this approach partially mitigates the first issue, the second issue remains unresolved. Specifically, during the mixing process, the adjugate of the instantaneous matrix, whose entries depend nonlinearly on closed-loop signals,  transforms the nonparametric uncertainty vector, making it infeasible to derive an explicit ultimate bound on the estimation error. An additional practical limitation is that the number of filters required for creating the instantaneous matrix scales with the number of unknown parameters. As the parameter dimension increases, selecting appropriate cut-off frequencies for each filter becomes increasingly tedious. To the best of the authors' knowledge, no systematic procedure currently exists for this task. \par

Existing methods that relax the persistent excitation condition to the finite excitation condition have been developed for solving regressor equations \cite{marino2022exponentially}, \cite{ortega2022new}, adaptive control with a known control effectiveness vector \cite{cho2017composite}, combined adaptive control with a known lower bound of the control effectiveness vector \cite{roy2017combined}, \cite{gerasimov2018relaxing}, and composite adaptive control in the presence of nonparametric uncertainties for Euler-Lagrange systems \cite{pan2018composite}. In combined adaptive control, the adaptation laws for estimating controller gains utilize both tracking and estimation errors \cite{duarte1989combined}. When memory terms are introduced into the parameter update law, additional terms appear in the Lyapunov-based stability analysis. These terms prevent the derivation of clean exponential decay conditions, thereby obstructing the establishment of explicit bounds on the closed-loop signals. To the best of the authors' knowledge, the design of a robust model reference adaptive controller with combined adaptation under the FE condition for a class of multi-input multi-output nonlinear systems with an unknown diagonal control effectiveness matrix and nonparametric uncertainties remains unaddressed. To address these challenges, we recently developed an algorithm for a single-input system without nonparametric uncertainties \cite{patel2025exponentially} that constructs a memory term yielding an identity coefficient matrix in the parameter estimation error dynamics under the FE condition, making the convergence rate independent of the excitation level. The present paper extends this approach to a significantly more general setting, as detailed below.\par

In this paper, we extend the approach to a class of multi-input multi-output systems with an unknown diagonal control effectiveness matrix and bounded nonparametric uncertainties. We propose an update law containing a modified memory term to estimate the unknown system parameters. The nonparametric uncertainty undergoes an orthogonal transformation and thus remains invariant under the designed memory term. The algorithm from \cite{patel2025exponentially} is modified to allow derivation of an explicit ultimate bound on the estimation error. Furthermore, we develop a combined adaptation law for estimating controller gains under the FE condition. Specifically, the controller gains are updated using the parameter estimates only when the FE condition is verified online. The proposed combined adaptive control scheme guarantees exponential decay of the tracking and estimation errors in a neighborhood of the origin, with an explicit ultimate bound under the FE condition. The exponential decay rate is independent of the entries of the memory matrix, owing to the modified algorithm design. 

The key contributions of this paper, beyond its conference version \cite{patel2025exponentially}, are as follows.

\begin{enumerate}
	\item The conference version addressed a single-input system without nonparametric uncertainties. This paper extends the framework to a class of multi-input multi-output nonlinear systems with an unknown diagonal control effectiveness matrix and investigates robustness to bounded nonparametric uncertainties.
	\item The algorithm for constructing the memory matrix is modified to enable the derivation of an explicit bound on the parameter estimation error.
	\item In the conference version, combined adaptation was driven by extracting the ideal system parameters under the FE condition. In this work, combined adaptation is driven by the online parameter estimates, guaranteeing exponential decay of the estimation errors to a neighborhood of the origin, characterized by an explicit ultimate bound.
\end{enumerate}

The main contributions of the proposed combined adaptive control scheme are summarized as follows.

\begin{enumerate}
	\item An algorithm is developed to construct a memory term by orthogonalizing the regressor via the Modified Gram-Schmidt procedure and applying a corresponding transformation to the output equation, preserving the linear-in-parameter structure and enabling explicit characterization of the transformed nonparametric uncertainty bound.
	\item A parameter update law is proposed based on the constructed memory term, which yields an identity coefficient matrix in the parameter estimation error dynamics under the finite excitation condition. This eliminates the need for time-varying adaptation gains. The parameter estimation error is shown to be uniformly ultimately bounded.
	\item A combined adaptation law is proposed under the finite excitation condition. The combined tracking and estimation errors are proven to decay exponentially to a neighborhood of the origin with an explicit bound. Unlike existing methods, the decay rate is independent of the excitation level of the regressor and depends solely on user-defined gains and system constants.
\end{enumerate}

The remainder of the article is organized as follows. Section \ref{section:Preliminaries} presents the notation, matrix identities, and definitions of the persistent and finite excitation conditions used throughout the paper. Section \ref{section:Problem_Formulation} formulates the system dynamics, reference model, and control objectives. Section \ref{section:Robust_Adaptive Control} reviews robust model reference adaptive control with the $\sigma$-modification and establishes its limitations as the baseline for the proposed approach. Section \ref{section:Main_Results} presents the main contributions: the orthogonal basis construction algorithm and parameter update law (Section $\ref{subsec:estimation_of_system_parameters}$), the combined adaptation law and closed-loop stability analysis (Section $\ref{subsec:combined_adaptation_in_finite_excitation_condition}$). Simulation results comparing the proposed scheme against concurrent learning, Memory Regressor and Extension (MRE), and DREM-based methods are presented in Section \ref{section:Simulation_Results}. Section \ref{section:conclusion} concludes the paper.

\section{Preliminaries}
\label{section:Preliminaries}

The following notations and definitions are used throughout this article. A scalar variable is denoted by a lowercase letter $m$, a vector by a lowercase boldface letter $\mathbf{m}$, and a matrix by an uppercase boldface letter $\mathbf{M}$. The Euclidean norm of a vector is denoted by $\|\mathbf{m}\|$, the $\mathcal{L}_{\infty}$ norm by $\|\mathbf{m}\|_{\infty}$, and the Frobenius norm of a matrix $\mathbf{M}$ by $\|\mathbf{M}\|_{F}$. For a square matrix $\mathbf{M}$, the minimum and maximum eigenvalues are denoted by $\lambda_{\text{min}}(\mathbf{M})$ and $\lambda_{\text{max}}(\mathbf{M})$, respectively, while $\det(\mathbf{M})$ and $\text{tr}(\mathbf{M})$ denote its determinant and trace. Several matrix identities and inequalities used in this article are summarized below.

\begin{enumerate}
	\item $\left\| \mathbf{M} \right\|_{F} = \left\| \text{vec}\left(\mathbf{M}\right) \right\|$ for any $\mathbf{M} \in \mathbb{R}^{n \times m}$, where $\text{vec}\left(\mathbf{M}\right)$ denotes the column-wise vectorization of the matrix $\mathbf{M}$.	
	\item $\mathbf{M} \mathbf{M}^{\top} = \mathbf{I}$ if and only if $\mathbf{M}$ is an orthogonal matrix, where $\mathbf{I}$ is the identity matrix.	
	\item $\text{tr}(\mathbf{M}_{1} + \mathbf{M}_{2}) = \text{tr}(\mathbf{M}_{1}) + \text{tr}(\mathbf{M}_{2})$ for any $\mathbf{M}_{1}, \mathbf{M}_{2} \in \mathbb{R}^{n \times n}$.	
	\item $\text{tr}(\boldsymbol{N} \mathbf{M}^{\top} \mathbf{M}) \geq \lambda_{\text{min}}(\boldsymbol{N}) \|\mathbf{M}\|_{F}^{2}$, where $\boldsymbol{N} \in \mathbb{R}^{m \times m}$ is a diagonal matrix and $\mathbf{M} \in \mathbb{R}^{n \times m}$.	
	\item $\text{tr}(\mathbf{m}_{1} \mathbf{m}_{2}^{\top}) = \mathbf{m}_{1}^{\top} \mathbf{m}_{2}$ for any $\mathbf{m}_{1} \in \mathbb{R}^{n}$ and $\mathbf{m}_{2} \in \mathbb{R}^{m}$.	
	\item $\left| \text{tr}(\mathbf{M} \boldsymbol{N}) \right| \leq \|\mathbf{M}\|_{F} \|\boldsymbol{N}\|_{F}$.	
	\item The trace of a product of matrices is invariant under cyclic permutations, i.e., $\text{tr}(\mathbf{A} \mathbf{B} \mathbf{C}) = \text{tr}(\mathbf{B} \mathbf{C} \mathbf{A}) = \text{tr}(\mathbf{C} \mathbf{A} \mathbf{B})$.
\end{enumerate}

\begin{definition}[\cite{yuan1977probing}]
	A bounded signal $\boldsymbol{\phi}(t) \in \mathbb{R}^{n}$ is said to be \emph{persistently exciting} if there exist positive constants $\gamma$ and $T$ such that, for every $t > 0$, there exists a sequence of $n$ time instances $t_{i} \in [t, t+T]$, $i \in \{1, 2, \dots, n\}$, satisfying
	\begin{equation}
		\left\| 
		\begin{bmatrix} 
			\boldsymbol{\phi}(t_{1}) & \boldsymbol{\phi}(t_{2}) & \cdots & \boldsymbol{\phi}(t_{n}) 
		\end{bmatrix}^{-1} 
		\right\| \leq \gamma_{\text{PE}},
		\label{eq:pe_definition}
	\end{equation}
	where $\gamma_{\text{PE}}$ represents the excitation level of the signal.
\end{definition}

\begin{definition}
	A bounded signal $\boldsymbol{\phi}(t) \in \mathbb{R}^{n}$ is said to be \emph{finitely exciting} if there exist positive constants $\gamma$ and $T$ such that at least one sequence of $n$ time instances $t_{i} \in [t, t+T]$, $i \in \{1, 2, \dots, n\}$, satisfies
	\begin{equation}
		\left\| 
		\begin{bmatrix} 
			\boldsymbol{\phi}(t_{1}) & \boldsymbol{\phi}(t_{2}) & \cdots & \boldsymbol{\phi}(t_{n}) 
		\end{bmatrix}^{-1} 
		\right\| \leq \gamma_{\text{FE}},
		\label{eq:fe_definition}
	\end{equation}
	where $\gamma_{\text{FE}}$ denotes the excitation level of the signal.
\end{definition}

\begin{remark}
FE is strictly weaker than PE as it requires richness over one finite window, not uniformly for all time. This condition is invoked in Assumption $\ref{assump:fe_condition}$ for an augmented regressor, formally defined in Section $\ref{section:Problem_Formulation}$.
\end{remark}

\section{Problem Formulation}
\label{section:Problem_Formulation}

In this section, we formalize the system dynamics and define the control problem. Consider a class of uncertain nonlinear systems \cite{lavretsky2009combined} described by the following dynamics:
\begin{equation}
	\dot{\mathbf{x}} = \mathbf{A} \mathbf{x} + \mathbf{B} \boldsymbol{\Lambda} \left\{ \mathbf{u} + \boldsymbol{\Theta}^{\top} \boldsymbol{\phi}(\mathbf{x}) \right\} + \boldsymbol{\xi},
	\label{eq:system_dynamics}
\end{equation}
where $\mathbf{x} \in \mathbb{R}^{n}$ is the measurable state vector, and $\mathbf{u}\in \mathbb{R}^{m}$ is the control input. The nonlinear matched parametric uncertainty $\boldsymbol{\Theta}^{\top} \boldsymbol{\phi}(\mathbf{x})$ consists of a known, continuously differentiable regressor vector $\boldsymbol{\phi}: \mathbb{R}^{n} \to \mathbb{R}^{p}$ and an unknown constant parameter matrix $\boldsymbol{\Theta} \in \mathbb{R}^{p \times m}$. 

The matrix $\mathbf{A} \in \mathbb{R}^{n \times n}$ and diagonal matrix $\boldsymbol{\Lambda} \in \mathbb{R}^{m \times m}$ are unknown, whereas $\mathbf{B} \in \mathbb{R}^{n \times m}$ is known\footnote{This structural assumption is standard in model reference adaptive control and is motivated by practical applications such as aerospace and mechanical systems where the input geometry is determined by the physical design and is therefore known, while the	system dynamics may be uncertain due to unmodeled effects or parameter variations.}. The term $\boldsymbol{\xi} \in \mathbb{R}^{n}$ denotes an unmatched nonparametric uncertainty that satisfies the bound $\left\| \boldsymbol{\xi}(t) \right\| \leq \overline{\xi}$ for all $t \in \mathbb{R}^{+}$.

\begin{assumption}
	The system dynamics in \eqref{eq:system_dynamics} is controllable; that is, the pair $(\mathbf{A}, \mathbf{B} \boldsymbol{\Lambda})$ is controllable, and $\mathbf{B}$ is a known full column rank matrix, i.e., $m \leq n$.
	\label{assump:controllable}
\end{assumption}

\begin{assumption}
	The matrix $\boldsymbol{\Lambda}$ is an unknown full-rank diagonal matrix with known signs of its diagonal elements. There exists a diagonal matrix $\boldsymbol{\Lambda}_{s}$, with diagonal entries $\pm1$, such that the product $\boldsymbol{\Lambda}_{s} \boldsymbol{\Lambda}$ is positive definite. Let $\underline{\lambda} > 0$ denote a known lower bound on $\lambda_{\text{min}} \left( \boldsymbol{\Lambda} \right)$.
	\label{assump:Lambda}
\end{assumption}

Our first control objective is to design a control law $\mathbf{u}$ for reference state tracking. The desired trajectory $\mathbf{x}_r$ is generated by the following reference model:
\begin{equation}
	\dot{\mathbf{x}}_{r} = \mathbf{A}_{r} \mathbf{x}_{r} + \mathbf{B}_{r} \mathbf{r},
	\label{eq:reference_system}
\end{equation}
where $\mathbf{A}_{r} \in \mathbb{R}^{n \times n}$ is a Hurwitz matrix, $\mathbf{x}_{r} \in \mathbb{R}^{n}$ is the reference state, and $\mathbf{r} \in \mathbb{R}^{m}$ is a bounded, piecewise-continuous reference input. The matrix $\mathbf{A}_{r}$ satisfies the Lyapunov equation $\mathbf{A}_{r}^{\top} \mathbf{P} + \mathbf{P} \mathbf{A}_{r} + \mathbf{Q} = \mathbf{0}$, where $\mathbf{P}, \mathbf{Q} \in \mathbb{R}^{n \times n}$ are symmetric and positive definite.

The classical feed-forward feedback-based control law is
\begin{equation}
	\mathbf{u} = \mathbf{K}_{x}^{\top} \mathbf{x} + \mathbf{K}_{r}^{\top} \mathbf{r} - \boldsymbol{\Theta}^{\top} \boldsymbol{\phi}(\mathbf{x}),
	\label{eq:control_input_ideal}
\end{equation}
where $\mathbf{K}_{x} \in \mathbb{R}^{n \times m}$ and $\mathbf{K}_{r} \in \mathbb{R}^{m \times m}$ are constant feedback and feed-forward gain matrices, respectively.

Define the tracking error $\mathbf{e} = \mathbf{x} - \mathbf{x}_r$. Using \eqref{eq:system_dynamics}, \eqref{eq:reference_system}, and \eqref{eq:control_input_ideal}, the tracking error dynamics is given by
\begin{equation}
	\dot{\mathbf{e}} = \mathbf{A}_{r} \mathbf{e} + \left( \mathbf{A} + \mathbf{B} \boldsymbol{\Lambda} \mathbf{K}_{x}^{\top} - \mathbf{A}_{r} \right) \mathbf{x} + \left( \mathbf{B} \boldsymbol{\Lambda} \mathbf{K}_{r}^{\top} - \mathbf{B}_{r} \right) \mathbf{r} + \boldsymbol{\xi}.
	\label{eq:tracking_error_ideal}
\end{equation}

The controller gains in the control law $\eqref{eq:control_input_ideal}$ are selected to satisfy the standard matching condition in model reference adaptive control as given in the following assumption.

\begin{assumption}
	There exist matrices $\mathbf{K}_{x}$ and $\mathbf{K}_{r}$ such that
	\begin{equation}
		\mathbf{A} + \mathbf{B} \boldsymbol{\Lambda} \mathbf{K}_{x}^{\top} = \mathbf{A}_{r}, \quad \quad \mathbf{B} \boldsymbol{\Lambda} \mathbf{K}_{r}^{\top} = \mathbf{B}_{r}.
		\label{eq:matching_condition}
	\end{equation}
	\label{assump:matching_condition}
\end{assumption}
The existence of such gains follows from the controllability of $\left( \mathbf{A}, \mathbf{B} \mathbf{\Lambda} \right)$ and standard results in adaptive control \cite{ioannou1996robust}. However, since $\mathbf{A}$, $\boldsymbol{\Lambda}$, and $\boldsymbol{\Theta}$ are unknown, the ideal gains cannot be computed directly. Therefore, the modified control law is
\begin{equation}
	\mathbf{u} = \hat{\mathbf{K}}_{x}^{\top} \mathbf{x} + \hat{\mathbf{K}}_{r}^{\top} \mathbf{r} - \hat{\boldsymbol{\Theta}}^{\top} \boldsymbol{\phi}(\mathbf{x}),
	\label{eq:control_input_modified}
\end{equation}
where $\hat{\mathbf{K}}_{x}$, $\hat{\mathbf{K}}_{r}$, and $\hat{\boldsymbol{\Theta}}$ are the respective estimates of the ideal gain matrices. Substituting \eqref{eq:control_input_modified} into the system dynamics and using the matching condition \eqref{eq:matching_condition}, the tracking error dynamics become
\begin{equation}
	\dot{\mathbf{e}} = \mathbf{A}_{r} \mathbf{e} + \mathbf{B} \boldsymbol{\Lambda} \tilde{\mathbf{K}}_{x}^{\top} \mathbf{x} + \mathbf{B} \boldsymbol{\Lambda} \tilde{\mathbf{K}}_{r}^{\top} \mathbf{r} - \mathbf{B} \boldsymbol{\Lambda} \tilde{\boldsymbol{\Theta}}^{\top} \boldsymbol{\phi}(\mathbf{x}) + \boldsymbol{\xi},
	\label{eq:tracking_error}
\end{equation}
where the estimation errors are defined as $\tilde{\mathbf{K}}_{x} = \hat{\mathbf{K}}_{x} - \mathbf{K}_{x}$, $\tilde{\mathbf{K}}_{r} = \hat{\mathbf{K}}_{r} - \mathbf{K}_{r}$, and $\tilde{\boldsymbol{\Theta}} = \hat{\boldsymbol{\Theta}} - \boldsymbol{\Theta}$.

The second control objective is to ensure convergence of the estimates $\hat{\mathbf{K}}_{x}$, $\hat{\mathbf{K}}_{r}$, and $\hat{\boldsymbol{\Theta}}$ to their ideal values under the finite excitation condition (Assumption $\ref{assump:fe_condition}$). This objective requires that the following regressor vector satisfies the finite excitation condition.
\begin{equation}
	\boldsymbol{\varphi}(\mathbf{x},\mathbf{u}) = 
	\begin{bmatrix}
		\mathbf{x}^{\top} & \mathbf{u}^{\top} & \boldsymbol{\phi}^{\top} \left( \mathbf{x} \right)
	\end{bmatrix}^{\top} \in \mathbb{R}^{n+m+p}.
	\label{eq:regressor_definition}
\end{equation}
Note that each block of $\boldsymbol{\varphi}$ is paired with a corresponding unknown quantity.
\begin{assumption}
	The regressor signal $\boldsymbol{\varphi}(\mathbf{x},\mathbf{u})$ satisfies the finite excitation (FE) condition that requires richness only over a finite time window rather than uniformly for all time.
	\label{assump:fe_condition}
\end{assumption}

Under Assumption $\ref{assump:fe_condition}$, there exist positive constants $\gamma$ and $T$ such that, within the interval $\left[ t, t+T \right]$, a sequence of $q$ time instances $\{t_1, \dots, t_q\}$ exists for which the regressor matrix satisfies:
\begin{equation}
	\left\| \begin{bmatrix}
		\boldsymbol{\varphi}(\mathbf{x}(t_{1}),\mathbf{u}(t_{1})) & \hdots & \boldsymbol{\varphi}(\mathbf{x}(t_{q}),\mathbf{u}(t_{q}))
	\end{bmatrix}^{-1} \right\| \leq \gamma,
	\label{eq:regressor_fe_condition}
\end{equation}
where $\gamma$ denotes the excitation level of the regressor.

\begin{remark}
The conditions in Assumption \ref{assump:Lambda} are physically motivated. In aerospace and mechanical systems, the signs of the diagonal entries of $\Lambda$ reflect the cause-and-effect relationship between actuator commands and system response, for example, a positive throttle command produces positive thrust, and a positive control surface deflection produces a moment in a known direction. These signs are therefore known from the system's physical design, even when the exact magnitudes are uncertain. Furthermore, the diagonal entries of $\Lambda$ are related to the effective mass and inertia of the system, quantities whose lower bounds are determined by maximum loading conditions and structural design limits. Consequently, a known positive lower bound on $\lambda_{\text{min}} \left( \boldsymbol{\Lambda} \right)$ is available from design specifications, even without precise knowledge of $\Lambda$ itself.
\end{remark}

The next section briefly reviews key results in robust adaptive control with $\sigma$-modification, which form the foundation of our main contribution. Specifically, we propose a novel update law and adaptation algorithm that guarantee convergence of the unknown parameters to their ideal values under Assumption $\ref{assump:fe_condition}$. These estimates are then incorporated into the control law to achieve the stated tracking and parameter convergence objectives.

\section{Review of Robust Adaptive Control with $\sigma$-Modification}
\label{section:Robust_Adaptive Control}

The control input $\eqref{eq:control_input_modified}$ requires estimation of the controller gains. We begin by reviewing an existing adaptive law that incorporates a robustness modification term, commonly referred to as the $\sigma$-modification \cite{narendra2012stable}. The adaptive update laws for estimating the controller gain, with the $\sigma$-modification, are

\begin{equation}
		\begin{bmatrix} \dot{\hat{\mathbf{K}}}_{x} & \dot{\hat{\mathbf{K}}}_{r} & \dot{\hat{\boldsymbol{\Theta}}} \end{bmatrix} = -\sigma \begin{bmatrix} \hat{\mathbf{K}}_{x} & \hat{\mathbf{K}}_{r} & \hat{\boldsymbol{\Theta}} \end{bmatrix} + \begin{bmatrix} - \boldsymbol{\Gamma}_{x} \mathbf{x} \mathbf{e}^{\top} \mathbf{P} \mathbf{B} \boldsymbol{\Lambda}_{s} & - \boldsymbol{\Gamma}_{r} \mathbf{r} \mathbf{e}^{\top} \mathbf{P} \mathbf{B} \boldsymbol{\Lambda}_{s} & \boldsymbol{\Gamma}_{\theta} \boldsymbol{\phi}\left( \mathbf{x} \right) \mathbf{e}^{\top} \mathbf{P} \mathbf{B} \boldsymbol{\Lambda}_{s} \end{bmatrix},
	\label{eq:controller_gain_update_law_gradient}
\end{equation}
where $\sigma>0$, and $\boldsymbol{\Gamma}_{x}$, $\boldsymbol{\Gamma}_{r}$, and $\boldsymbol{\Gamma}_{\theta}$ are positive definite adaptation gain matrices. To streamline the presentation, we assume $\sigma=1$ and the gain matrices to be identity matrices in the remainder of the paper. The following proposition summarizes key properties of the robust adaptive scheme with the $\sigma$-modification.

\begin{proposition}
	Consider the class of uncertain nonlinear systems described by $\eqref{eq:system_dynamics}$, driven by the control input $\eqref{eq:control_input_modified}$, where the controller gains are updated using the adaptation laws $\eqref{eq:controller_gain_update_law_gradient}$. The closed-loop system exhibits the following properties.
	\begin{enumerate}
		\item All the closed-loop signals remain uniformly ultimately bounded.
		\item The ultimate bound depends on the upper-bound of the nonparametric uncertainty $\overline{\xi}$, the ideal controller gains $\mathbf{K}_{x}$ and $\mathbf{K}_{r}$, and the system parameter $\boldsymbol{\Theta}$.
		\item The ultimate bound does not decay to zero even in the absence of nonparametric uncertainty $\boldsymbol{\xi}$. 
	\end{enumerate}  
	\label{proposition:direct_adaptive}
\end{proposition}
\renewcommand\qedsymbol{$\blacksquare$}
\begin{proof} The proof follows from a standard Lyapunov argument using the candidate $V_e$ defined in Appendix A; the key step is that the $\sigma$-modification introduces a residual term that persists even when $\boldsymbol{\xi} = 0$. \end{proof}

While the $\sigma$-modification guarantees boundedness of the closed-loop signals, it results in a non-vanishing steady-state error in the parameter estimates, even under persistent excitation condition. This residual error persists because the $\sigma$-modification introduces a bias term $\sigma \hat{\mathbf{K}}_{x}$ in the parameter update law that cannot be eliminated by excitation alone. It pulls the estimates toward zero rather than toward their ideal values. As a result, the ultimate bound does not decay to zero even in the absence of nonparametric uncertainty $\boldsymbol{\xi}$, as established in property 3 of Proposition $\ref{proposition:direct_adaptive}$ and detailed in Appendix A.

To overcome this limitation, we propose a combined adaptive law in Section $\ref{subsec:combined_adaptation_in_finite_excitation_condition}$ $\eqref{eq:controller_gain_update_law_proposed}$. This law selectively activates parameter adaptation using a memory term only when finite excitation is detected online, that is, when the regressor matrix becomes full rank over a finite time window. When this condition is not yet satisfied, the adaptation reverts to the standard error-minimizing update rule with $\sigma$-modification, preserving robustness to nonparametric uncertainty. Once the excitation condition is met, the combined law drives the parameter estimates toward their ideal values, achieving the convergence that the $\sigma$-modification alone cannot provide.

\section{Main Results}
\label{section:Main_Results}

The combined adaptation involves updating the controller gain estimates using both tracking errors and estimation errors. The estimation errors are derived from the matching condition \eqref{eq:matching_condition}, which depends on the unknown system matrices. To facilitate the combined adaptation, the ideal controller gains are expressed using the matching condition \eqref{eq:matching_condition} as given below.
\begin{equation}
	\boldsymbol{\Lambda} \mathbf{K}_{x}^{\top} = \mathbf{B}^{\dagger} \left( \mathbf{A}_{r} - \mathbf{A} \right), \quad  
	\boldsymbol{\Lambda} \mathbf{K}_{r}^{\top} = \mathbf{B}^{\dagger} \mathbf{B}_{r},
	\label{eq:matching_condition_ideal_gains}
\end{equation}
where $\mathbf{B}^{\dagger}$ denotes the Moore–Penrose inverse of the matrix $\mathbf{B}$. These expressions require knowledge of the ideal values of $\boldsymbol{\Lambda}$ and $\mathbf{B}^{\dagger} \mathbf{A}$. To estimate these values, the system dynamics is expressed as follows using \eqref{eq:regressor_definition}.
\begin{equation}
	\begin{aligned}
		\mathbf{B}^{\dagger} \dot{\mathbf{x}} &= \mathbf{B}^{\dagger} \mathbf{A} \mathbf{x} + \boldsymbol{\Lambda} \left\{ \mathbf{u} + \boldsymbol{\Theta}^{\top} \boldsymbol{\phi}(\mathbf{x}) \right\} + \mathbf{B}^{\dagger} \boldsymbol{\xi}, \\
		\mathbf{y} &= \underbrace{\begin{bmatrix} \mathbf{B}^{\dagger} \mathbf{A} & \boldsymbol{\Lambda} & \boldsymbol{\Lambda} \boldsymbol{\Theta}^{\top} \end{bmatrix}}_{\mathbf{W}^{\top}} 
		\boldsymbol{\varphi}(\mathbf{x},\mathbf{u}) + \boldsymbol{\xi}_{y}
	\end{aligned}
	\label{eq:system1}
\end{equation}
where $\mathbf{y} := \mathbf{B}^{\dagger} \dot{\mathbf{x}}$, $\boldsymbol{\xi}_{y} := \mathbf{B}^{\dagger} \boldsymbol{\xi}$, and hence $\mathbf{y}, \boldsymbol{\xi}_{y} \in \mathbb{R}^{m}$. The matrix $\mathbf{W} \in \mathbb{R}^{(n + m + p) \times m}$. Moreover, the transformed uncertainty remains bounded as given below.
\begin{equation}
	\left\| \boldsymbol{\xi}_{y} \right\| \leq \overline{\xi}_{y}, \quad 
	\overline{\xi}_{y} := \left\| \mathbf{B}^{\dagger} \right\| \, \overline{\xi}.
	\label{eq:upper_bound_uncertainty}
\end{equation}
Since $\mathbf{B}$ is a known constant matrix, $\overline{\xi}_{y}$ is a computable upper bound on the transformed nonparametric uncertainty. Thus, the boundedness assumption on $\boldsymbol{\xi}$ is preserved under the transformation \eqref{eq:system1}. Re-writing the transformed system dynamics \eqref{eq:system1} in a linear-in-parameter form

\begin{equation}
	\mathbf{y} = 
	\underbrace{
		\begin{bmatrix}
			\boldsymbol{\varphi}^{\top} & \mathbf{0}^{\top} & \cdots & \mathbf{0}^{\top} \\
			\mathbf{0}^{\top} & \boldsymbol{\varphi}^{\top} & \cdots & \mathbf{0}^{\top} \\
			\vdots & \vdots & \ddots & \vdots \\
			\mathbf{0}^{\top} & \mathbf{0}^{\top} & \cdots & \boldsymbol{\varphi}^{\top}
		\end{bmatrix}
	}_{\boldsymbol{\Phi}^{\top}(\mathbf{x}, \mathbf{u})} 
	\underbrace{ \mathrm{vec}(\mathbf{W}) }_{\mathbf{w}} + \boldsymbol{\xi}_{y},
	\label{eq:system_lip}
\end{equation}
where $\mathrm{vec}(\mathbf{W})$ denotes the column-wise vectorization of $\mathbf{W}$, resulting in $\mathbf{w} \in \mathbb{R}^{mq}$ with $q := n + m + p$. Thus, $\boldsymbol{\Phi} \in \mathbb{R}^{mq \times m}$, and $\mathbf{0}$ denotes a zero vector of length $q$. The regressor $\boldsymbol{\Phi}$ captures the measurable components of \eqref{eq:system1}, while the unknown parameters are grouped in $\mathbf{w}$. The recovery of the individual controller gain estimates $\hat{\mathbf{K}}_{x}$, $\hat{\mathbf{K}}_{r}$, and $\hat{\boldsymbol{\Theta}}$ from the system parameter estimate $\hat{\mathbf{w}}$ is addressed in Section $\ref{subsec:combined_adaptation_in_finite_excitation_condition}$, where the matching condition is used to extract $\hat{\boldsymbol{\Lambda}}$ and $\hat{\mathbf{A}}$ from $\hat{\mathbf{w}}$ and subsequently compute the gain estimates.

Since $\mathbf{y}$ involves $\dot{\mathbf{x}}$ and is not directly measurable, the system \eqref{eq:system1} is passed through a first-order stable filter. Let $\boldsymbol{\varphi}_f(t,\mathbf{x},\mathbf{u})$ and $\mathbf{y}_f(t, \mathbf{x})$ denote the filtered versions of $\boldsymbol{\varphi}(\mathbf{x}, \mathbf{u})$ and $\mathbf{y}(t)$, respectively, obtained as follows using a stable filter $H(s) = \dfrac{f}{s + f}$ with $f > 0$.
\begin{equation}
	\begin{aligned}
		\dot{\boldsymbol{\varphi}}_{f} &= -f \boldsymbol{\varphi}_{f} + f \boldsymbol{\varphi}, \\
		\dot{\mathbf{x}}_{f} &= -f \mathbf{x}_{f} + f \mathbf{x}, \\
		\mathbf{y}_{f} &= \mathbf{B}^{\dagger} \left( f \mathbf{x} - f e^{-f t} \mathbf{x}(t_0) - f \mathbf{x}_{f} \right).
	\end{aligned}
	\label{eq:filter_solution}
\end{equation}
For notational simplicity, we omit the explicit arguments of the filtered signals in subsequent expressions, except when necessary. Assuming initial conditions $\boldsymbol{\varphi}_{f}(0) = 0$ and $\mathbf{y}_{f}(0) = 0$, the filtered system is
\begin{equation}
	\begin{aligned}
		\mathbf{y}_{f} &= \mathbf{W}^{\top} \boldsymbol{\varphi}_{f} + \boldsymbol{\xi}_{f} \\
		&= \boldsymbol{\Phi}_{f}^{\top} \mathbf{w} + \boldsymbol{\xi}_{f}, 
	\end{aligned}
	\quad \text{where } \quad
	\boldsymbol{\Phi}_{f}^{\top} := 
	\begin{bmatrix}
		\boldsymbol{\varphi}_{f}^{\top} & \mathbf{0}^{\top} & \cdots & \mathbf{0}^{\top} \\
		\mathbf{0}^{\top} & \boldsymbol{\varphi}_{f}^{\top} & \cdots & \mathbf{0}^{\top} \\
		\vdots & \vdots & \ddots & \vdots \\
		\mathbf{0}^{\top} & \mathbf{0}^{\top} & \cdots & \boldsymbol{\varphi}_{f}^{\top}
	\end{bmatrix}
	\in \mathbb{R}^{m \times mq}.
	\label{eq:filtered_system}
\end{equation}
Here, $\boldsymbol{\xi}_{f}$ denotes the filtered version of $\boldsymbol{\xi}_{y}$, and $\boldsymbol{\Phi}_{f}$ represents the filtered regressor matrix. Next, we propose an update law to estimate the unknown system parameters contained in $\mathbf{w}$ and subsequently analyze the stability properties of the resulting identification system.

\subsection{Estimation of System Parameters}
\label{subsec:estimation_of_system_parameters}
Let the regressor vector $\boldsymbol{\varphi}$ satisfy the finite excitation condition (see \textit{Assumption $\ref{assump:fe_condition}$}). Then, the filtered regressor vector $\boldsymbol{\varphi}_{f}$ also satisfies the finite excitation condition [\textit{Lemma 6.8}, \cite{narendra2012stable}], implying that a full rank matrix can be constructed by storing $\boldsymbol{\varphi}_{f}$ at regularly sampled time instants. Let $\boldsymbol{\Phi}_{fc} = \begin{bmatrix} \boldsymbol{\varphi}_{f} \left(t_{1}\right) & \hdots & \boldsymbol{\varphi}_{f} \left(t_{q}\right) \end{bmatrix} \in \mathbb{R}^{q \times q}$ and $\mathbf{Y}_{fc} = \begin{bmatrix} \mathbf{y}_{f} \left(t_{1} \right) & \hdots & \mathbf{y}_{f} \left(t_{q}\right) \end{bmatrix} \in \mathbb{R}^{m \times q}$ be the matrices formed by collecting $\boldsymbol{\varphi}_{f}$ and $\mathbf{y}_{f}$, respectively at $q$ such time instances. These matrices preserve the linear-in-parameter form.
\begin{equation}
	\mathbf{Y}_{fc} = \mathbf{W}^{\top} \boldsymbol{\Phi}_{fc} + \boldsymbol{\Xi}_{fc},
	\label{eq:filtered_system_collection}
\end{equation}
where $\boldsymbol{\Xi}_{fc} = \begin{bmatrix}\boldsymbol{\xi}_{f} \left(t_{1}\right) & \hdots & \boldsymbol{\xi}_{f} \left(t_{q}\right)\end{bmatrix}$ is the corresponding nonparametric uncertainty matrix. Vectorizing the transposes, define $\mathbf{y}_{fc} = vec\left(\mathbf{Y}_{fc}^{\top}\right) \in \mathbb{R}^{mq}$ and $\boldsymbol{\xi}_{fc} = vec\left(\boldsymbol{\Xi}_{fc}^{\top}\right) \in \mathbb{R}^{mq}$. Then, we obtain
\begin{equation}
	\mathbf{y}_{fc} = \boldsymbol{\Psi}_{fc}^{\top} \mathbf{w} + \boldsymbol{\xi}_{fc}, \quad
	\boldsymbol{\Psi}_{fc}^{\top} = \begin{bmatrix} \boldsymbol{\Phi}_{fc}^{\top} & \mathbf{0}_{q \times q} & \hdots & \mathbf{0}_{q \times q} \\ \mathbf{0}_{q \times q} & \boldsymbol{\Phi}_{fc}^{\top} & \hdots & \mathbf{0}_{q \times q} \\ \vdots & \vdots & \ddots & \vdots \\ \mathbf{0}_{q \times q} & \mathbf{0}_{q \times q} & \hdots & \boldsymbol{\Phi}_{fc}^{\top} \end{bmatrix} \in \mathbb{R}^{mq \times mq}.
	\label{eq:filtered_collection_lip}
\end{equation}
We now convert the full-rank matrix $\boldsymbol{\Psi}_{fc}$ into an orthogonal matrix using \textit{Algorithm} $\ref{alg:identification}$. Since $\boldsymbol{\Psi}_{fc}$ is block diagonal with $\boldsymbol{\Phi}_{fc}$ along its diagonal, it suffices to orthogonalize $\boldsymbol{\Phi}_{fc}$ into $\boldsymbol{\Phi}_{fb}$ and replicate it along the diagonal to construct $\boldsymbol{\Psi}_{fb}$. Correspondingly, we transform $\mathbf{Y}_{fc}$ into $\mathbf{Y}_{fb}$. Define $\boldsymbol{\Phi}_{fb} = \begin{bmatrix} \boldsymbol{\varphi}_{fb_{1}} & \boldsymbol{\varphi}_{fb_{2}} & \hdots & \boldsymbol{\varphi}_{fb_{q}} \end{bmatrix}$ and $\mathbf{Y}_{fb} = \begin{bmatrix} \mathbf{y}_{fb_{1}} & \mathbf{y}_{fb_{2}} & \hdots & \mathbf{y}_{fb_{q}} \end{bmatrix}$. The first orthonormal column $\boldsymbol{\varphi}_{fb_{1}}$ of $\boldsymbol{\Phi}_{fb}$ is computed as below.
\begin{equation}
	\boldsymbol{\varphi}_{fb_{1}} = \dfrac{\boldsymbol{\varphi}_{f}\left(t_{1}\right)}{\left\| \boldsymbol{\varphi}_{f}\left(t_{1}\right) \right\|}, \quad \left\| \boldsymbol{\varphi}_{f}\left(t_{1}\right) \right\| \geq \delta_{1}
	\label{eq:orthonormal_column_1}
\end{equation}
where $\delta_{1} >0$ is a design parameter. The remaining orthonormal columns are then computed recursively as
\begin{equation}
	\mathbf{v}_{l} = \mathbf{v}_{l-1} - \boldsymbol{\varphi}_{fb_{l-1}}^{\top} \mathbf{v}_{l-1}  \boldsymbol{\varphi}_{fb_{l-1}}, \quad \boldsymbol{\varphi}_{fb_{k}} = \dfrac{\mathbf{v}_{k}}{\left\| \mathbf{v}_{k} \right\|},
	\label{eq:orthonormal_column_i}	
\end{equation}
with $l:2 \to k$, $\mathbf{v}_{1} = \dfrac{\boldsymbol{\varphi}_{f}\left(t_{k}\right)}{\left\| \boldsymbol{\varphi}_{f}\left(t_{k}\right) \right\|}$, and $k= \left\{2,3,\hdots,q \right\}$. This computation is performed conditionally to ensure robustness against nonparametric uncertainties. \par

To maintain the linear-in-parameter form, the following transformation is proposed whose steps are detailed in \textit{Algorithm} $\ref{alg:identification}$.
\begin{equation}
	\begin{aligned}
		\mathbf{y}_{fb_{1}} =& \dfrac{\mathbf{y}_{f} \left(t_{1} \right)}{\left\| \boldsymbol{\varphi}_{f}\left(t_{1}\right) \right\|}, \\
		\mathcal{Y}_{l} =& \mathcal{Y}_{{l-1}} - \boldsymbol{\varphi}_{fb_{l-1}}^{\top} \mathbf{v}_{l-1} \mathbf{y}_{fb_{l-1}}, \quad  \mathbf{y}_{fb_{k}} = \dfrac{1}{\left\| \mathbf{v}_{k} \right\|} \mathcal{Y}_{k},
	\end{aligned}
	\label{eq:orthonormal_y}
\end{equation}
with $l:2 \to k$, $\mathcal{Y}_{1} = \dfrac{\mathbf{y}_{f} \left( t_{k} \right)}{\left\| \boldsymbol{\varphi}_{f}\left(t_{k}\right) \right\|}$, $k= \left\{2,3,\hdots,q \right\}$, and $\mathbf{v}_{l}$ is computed from $\eqref{eq:orthonormal_column_i}$. The proposed orthogonalization and transformation of $\mathbf{Y}_{fc}$ using \textit{Algorithm $\ref{alg:identification}$} retain the linear-in-parameter structure. This is formally established in the subsequent theorem, where an explicit expression for  the transformed nonparametric uncertainty is derived in terms of known quantities and $\boldsymbol{\xi}_{f}$. \par

\begin{algorithm}
	\begin{algorithmic}[1]
		\State Initialize $\delta_{1}$, $\delta_{2}$
		\State Set q $\gets$ Number of Unknown Parameters
		\State Set k $\gets 1$
		\If{k $\leq$ q}
		\If{$\left\| \boldsymbol{\varphi}_{f} \right\| \geq \delta_{1}$}
		\State $\mathbf{v}_{1} \gets $ $\dfrac{\boldsymbol{\varphi}_{f}}{\left\| \boldsymbol{\varphi}_{f} \right\|}$  and $\mathcal{Y}_{1} \gets $ $\dfrac{\mathbf{y}_{f}}{\left\| \boldsymbol{\varphi}_{f} \right\|}$
		\For{l=1:k-1}
		\State $\mathbf{v}_{l+1} = \mathbf{v}_{l} - \boldsymbol{\varphi}_{fb_{l}}^{\top} \mathbf{v}_{l}  \boldsymbol{\varphi}_{fb_{l}}$
		\EndFor
		\If{ $\left\| \mathbf{v}_{k} \right\| \geq \delta_{2} $}
		\State Compute $\boldsymbol{\varphi}_{fb_{k}}$ and $\mathbf{y}_{fb_{k}}$ using $\eqref{eq:orthonormal_column_i}$ and $\eqref{eq:orthonormal_y}$ respectively
		\State k=k+1
		\EndIf
		\EndIf
		\EndIf
		\caption{Orthogonal Basis Generation for $\boldsymbol{\Psi}_{fb}$ and $\mathbf{y}_{fb}$ in $\eqref{eq:filtered_basis}$}
		\label{alg:identification}
	\end{algorithmic}
\end{algorithm}

\begin{theorem}
	Consider the filtered regression model $\eqref{eq:filtered_system}$. Applying the Modified Gram-Schmidt procedure using the \textit{Algorithm} $\ref{alg:identification}$ yields the following properties.
	\begin{enumerate}
		\item The transformed matrices $\mathbf{Y}_{fb}$ and $\boldsymbol{\Phi}_{fb}$ retain the following linear-in-parameter structure.
		\begin{equation}
			\begin{aligned}
				\mathbf{Y}_{fb} &= \mathbf{W}^{\top} \boldsymbol{\Phi}_{fb} + \boldsymbol{\Xi}_{fb}, \\
				\mathbf{y}_{fb} &= \boldsymbol{\Psi}_{fb}^{\top} \mathbf{w} + \boldsymbol{\xi}_{fb},
			\end{aligned}
			\label{eq:filtered_basis}
		\end{equation}
		where $\boldsymbol{\Psi}_{fb}^{\top} \in \mathbb{R}^{mq \times mq}$ is block-diagonal with blocks $\boldsymbol{\Phi}_{fb}^{\top}$, $\mathbf{y}_{fb}= \mathrm{vec}\left(\mathbf{Y}_{fb}^{\top}\right) \in \mathbb{R}^{mq}$, and $\boldsymbol{\xi}_{fb} = \mathrm{vec}\left(\boldsymbol{\Xi}_{fb}^{\top}\right) \in \mathbb{R}^{mq}$.
		\item Each column $\boldsymbol{\xi}_{fb_{k}}$ of the transformed nonparametric uncertainty matrix $\boldsymbol{\Xi}_{fb} = \begin{bmatrix} \boldsymbol{\xi}_{fb_{1}} & \boldsymbol{\xi}_{fb_{2}} & \cdots & \boldsymbol{\xi}_{fb_{q}} \end{bmatrix}$ is 
		\begin{equation}
			\boldsymbol{\xi}_{fb_{k}} = \dfrac{1}{\left\| \mathbf{v}_{k} \right\|} \left( \dfrac{\boldsymbol{\xi}_{f} \left(t_{k}\right)}{\left\| \boldsymbol{\varphi}_{f} \left(t_{k}\right) \right\|} - \sum_{j=1}^{k-1} \boldsymbol{\varphi}_{fb_{j}}^{\top} \mathbf{v}_{j} \boldsymbol{\xi}_{fb_{j}} \right),
			\label{eq:lip_uncertain_expression}
		\end{equation}
		where $\mathbf{v}_j$ is the $j$-th intermediate vector in the Modified Gram–Schmidt procedure.
		\item The vectorized uncertainty $\boldsymbol{\xi}_{fb}$ satisfies the following bound.
		\begin{equation}
			\left\| \boldsymbol{\xi}_{fb} \right\| \leq \overline{\xi}_{fb}, \quad \text{where} \quad \overline{\xi}_{fb} = \left( \dfrac{1}{\delta_{1}} + \dfrac{2 \left( \left(1+\delta_{2}\right)^{q-1} - \delta_{2}^{q-1} \right)}{\delta_{1} \delta_{2}^{q-1}} \right) \overline{\xi}_{y},
			\label{eq:filtered_bound}
		\end{equation}
		and $\overline{\xi}_{y} = \left\| \mathbf{B}^{\dagger} \right\| \, \overline{\xi}$ [see \eqref{eq:upper_bound_uncertainty}].
	\end{enumerate}	
	\label{theorem:algorithm1_lip}
\end{theorem}

\begin{proof}
	The proof is established using mathematical induction and recursive substitution. To establish the preservation of the linear-in-parameter structure, we apply mathematical induction over the Modified Gram-Schmidt iterations. 
	
	Let the linear-in-parameter relation for the transformed vectors at $l-$th iteration, $\forall \, k\in \left\{1,2,\hdots,q\right\}$ be denoted as follows.
	\begin{equation*}
		P\left(l\right): \mathbf{y}_{fb_{l}} = \mathbf{W}^{\top} \boldsymbol{\varphi}_{fb_{l}} + \boldsymbol{\xi}_{fb_{l}},
	\end{equation*}
	where $l=\left\{1,2,\hdots,k\right\}$. For $l=1$, using $\eqref{eq:orthonormal_column_1}$ and $\eqref{eq:orthonormal_y}$
	\begin{equation*}
		\boldsymbol{\varphi}_{fb_{1}} = \dfrac{\boldsymbol{\varphi}_{f}\left(t_{1}\right)}{\left\| \boldsymbol{\varphi}_{f}\left(t_{1}\right) \right\|}, \quad \mathbf{y}_{fb_{1}} = \dfrac{\mathbf{y}_{f} \left(t_{1}\right)}{\left\| \boldsymbol{\varphi}_{f}\left(t_{1}\right) \right\|}.
	\end{equation*}
	From $\eqref{eq:filtered_system}$, substituting $\mathbf{y}_{f} = \mathbf{W}^{\top} \boldsymbol{\varphi}_{f} + \boldsymbol{\xi}_{f}$ in $\mathbf{y}_{fb_{1}}$.
	\begin{equation*}
	\begin{aligned}
		\mathbf{y}_{fb_{1}} &= \dfrac{\mathbf{y}_{f} \left(t_{1}\right)}{\left\| \boldsymbol{\varphi}_{f}\left(t_{1}\right) \right\|} = \dfrac{\mathbf{W}^{\top} \boldsymbol{\varphi}_{f} \left(t_{1}\right) + \boldsymbol{\xi}_{f} \left(t_{1}\right)}{\left\| \boldsymbol{\varphi}_{f}\left(t_{1}\right) \right\|}, \\
			&= \mathbf{W}^{\top} \boldsymbol{\varphi}_{fb_{1}}	+ \boldsymbol{\xi}_{fb_{1}}, \quad \boldsymbol{\xi}_{fb_{1}} = \dfrac{\boldsymbol{\xi}_{f} \left(t_{1}\right)}{\left\| \boldsymbol{\varphi}_{f}\left(t_{1}\right) \right\|}.
	\end{aligned}
	\end{equation*}
	Therefore, $P\left(1\right)$ is true.
	
	Assume $P\left(l\right)$ is true, that is, $\mathbf{y}_{fb_{l}} = \mathbf{W}^{\top} \boldsymbol{\varphi}_{fb_{l}} + \boldsymbol{\xi}_{fb_{l}}$. Then, from $\eqref{eq:orthonormal_column_i}$ and $\eqref{eq:orthonormal_y}$, it follows that $\mathcal{Y}_{l} = \mathbf{W}^{\top} \mathbf{v}_{l} + \left\| \mathbf{v}_{l} \right\| \boldsymbol{\xi}_{fb_{l}}$. If the relation $\mathcal{Y}_{l+1} = \mathbf{W}^{\top} \mathbf{v}_{l+1} + \left\| \mathbf{v}_{l+1} \right\| \boldsymbol{\xi}_{fb_{l+1}}$ holds, it infers that $\mathbf{y}_{fb_{l+1}} = \mathbf{W}^{\top} \boldsymbol{\varphi}_{fb_{l+1}} + \boldsymbol{\xi}_{fb_{l+1}}$. Thus, $P\left(l+1\right)$ is true. From $\eqref{eq:orthonormal_column_i}$ and $\eqref{eq:orthonormal_y}$,
	\begin{equation*}
		\mathbf{v}_{l+1} = \mathbf{v}_{l} - \boldsymbol{\varphi}_{fb_{l}}^{\top} \mathbf{v}_{l}  \boldsymbol{\varphi}_{fb_{l}}, \quad
		\mathcal{Y}_{l+1} = \mathcal{Y}_{{l}} - \boldsymbol{\varphi}_{fb_{l}}^{\top} \mathbf{v}_{l} \mathbf{y}_{fb_{l}}.
	\end{equation*}
	Substituting $\mathcal{Y}_{l} = \mathbf{W}^{\top} \mathbf{v}_{l} + \left\| \mathbf{v}_{l} \right\| \boldsymbol{\xi}_{f_{l}} \left(t_{k}\right)$ and $\mathbf{y}_{fb_{l}} = \mathbf{W}^{\top} \boldsymbol{\varphi}_{fb_{l}} + \boldsymbol{\xi}_{fb_{l}}$ in $\mathcal{Y}_{l+1}$.
	\begin{equation*}
		\mathcal{Y}_{l+1} = \mathbf{W}^{\top} \mathbf{v}_{l} + \left\| \mathbf{v}_{l} \right\| \boldsymbol{\xi}_{f_{l}} \left(t_{k}\right) - \boldsymbol{\varphi}_{fb_{l}}^{\top} \mathbf{v}_{l} \left(\mathbf{W}^{\top} \boldsymbol{\varphi}_{fb_{l}} + \boldsymbol{\xi}_{fb_{l}} \right),
	\end{equation*}
	where $\boldsymbol{\xi}_{f_{l}} \left(t_{k}\right)$ is the recursive transformation of $\boldsymbol{\xi}_{f} \left(t_{k}\right)$ at $l-$th iteration $\forall \, k$, and $\boldsymbol{\varphi}_{fb_{l}}^{\top} \mathbf{v}_{l}$ is a scalar.
	Pre-multiplying with $\dfrac{1}{\left\| \mathbf{v}_{l+1} \right\|}$.
	\begin{equation*}
		\mathbf{y}_{fb_{l+1}} = \mathbf{W}^{\top} \boldsymbol{\varphi}_{fb_{l+1}} + \dfrac{\left\| \mathbf{v}_{l} \right\|}{\left\| \mathbf{v}_{l+1} \right\|} \boldsymbol{\xi}_{f_{l}} \left(t_{k}\right) - \dfrac{\boldsymbol{\varphi}_{fb_{l}}^{\top} \mathbf{v}_{l}}{\left\| \mathbf{v}_{l+1} \right\|} \boldsymbol{\xi}_{fb_{l}},
	\end{equation*}
	where,
	\begin{equation}
		\boldsymbol{\xi}_{fb_{l+1}} = \dfrac{\left\| \mathbf{v}_{l} \right\|}{\left\| \mathbf{v}_{l+1} \right\|} \boldsymbol{\xi}_{f_{l}} \left(t_{k}\right) - \dfrac{\boldsymbol{\varphi}_{fb_{l}}^{\top} \mathbf{v}_{l}}{\left\| \mathbf{v}_{l+1} \right\|} \boldsymbol{\xi}_{fb_{l}}.
		\label{eq:nonparametric_uncertainty_lth}
	\end{equation}
	Thus $P\left(l+1\right)$ is also true when $P\left(l\right)$ is true. Hence, by the principle of  induction, $P(l)$ holds $\forall \, l \in \left\{1,2,\hdots,k\right\}$, $\forall \, k\in \left\{1,2,\hdots,q\right\}$. This completes the proof of the first statement.
	
	Now, consider the transformed nonparametric uncertainty $\boldsymbol{\xi}_{f_{l}} \left(t_{k}\right)$ at the $l-$th iteration, $\forall \, k \in \left\{1,2,\hdots,q\right\}$, defined recursively as follows.
	\begin{equation*}
		P(l): \boldsymbol{\xi}_{f_{l}} \left(t_{k}\right) = \dfrac{1}{\left\| \mathbf{v}_{l} \right\|} \left( \dfrac{\boldsymbol{\xi}_{f} \left(t_{k}\right)}{\left\| \boldsymbol{\varphi}_{f} \left(t_{k}\right) \right\|} - \sum_{j=1}^{l-1} \boldsymbol{\varphi}_{fb_{j}}^{\top} \mathbf{v}_{j} \boldsymbol{\xi}_{fb_{j}} \right),
	\end{equation*}
	where $\boldsymbol{\xi}_{f_{l}} \left(t_{k}\right)$ is the recursive transformation of $\boldsymbol{\xi}_{f} \left(t_{k}\right)$ at $l-$th iteration $\forall \, k \in \left\{1,2,\hdots,q\right\}$. For $l=1$, $P\left(1\right): \boldsymbol{\xi}_{f_{1}} \left(t_{k}\right)=\dfrac{\boldsymbol{\xi}_{f}\left(t_{k}\right)}{\left\| \boldsymbol{\varphi}_{f} \left(t_{1}\right) \right\|}$, which is true as derived earlier. Assume $P\left(l\right)$ holds and evaluate $P\left(l+1\right)$. Using $\eqref{eq:nonparametric_uncertainty_lth}$, the transformed nonparametric uncertainty $\boldsymbol{\xi}_{f_{l+1}}$ at $\left(l+1\right)-$th iteration, $\forall \, k \in \left\{1,2,\hdots,q\right\}$ is as follows.
	\begin{equation*}
		\boldsymbol{\xi}_{f_{l+1}} \left(t_{k}\right) = \dfrac{\left\| \mathbf{v}_{l} \right\|}{\left\| \mathbf{v}_{l+1} \right\|} \boldsymbol{\xi}_{f_{l}} \left(t_{k}\right) - \dfrac{\boldsymbol{\varphi}_{fb_{l}}^{\top} \mathbf{v}_{l}}{\left\| \mathbf{v}_{l+1} \right\|} \boldsymbol{\xi}_{fb_{l}}.
	\end{equation*}
	Substituting $\boldsymbol{\xi}_{f_{l}}\left(t_{k}\right)$ in $\boldsymbol{\xi}_{f_{l+1}}\left(t_{k}\right)$.
	\begin{equation*}
	\begin{aligned}
		\boldsymbol{\xi}_{f_{l+1}} \left(t_{k}\right) &= \dfrac{1}{\left\| \mathbf{v}_{l+1} \right\|} \left( \dfrac{\boldsymbol{\xi}_{f} \left(t_{k}\right)}{\left\| \boldsymbol{\varphi}_{f} \left(t_{k}\right) \right\|} - \sum_{j=1}^{l-1} \boldsymbol{\varphi}_{fb_{j}}^{\top} \mathbf{v}_{j} \boldsymbol{\xi}_{fb_{j}} \right) - \dfrac{\boldsymbol{\varphi}_{fb_{l}}^{\top} \mathbf{v}_{l}}{\left\| \mathbf{v}_{l+1} \right\|} \boldsymbol{\xi}_{fb_{l}}, \\
			&= \dfrac{1}{\left\| \mathbf{v}_{l+1} \right\|} \left( \dfrac{\boldsymbol{\xi}_{f} \left(t_{k}\right)}{\left\| \boldsymbol{\varphi}_{f} \left(t_{k}\right) \right\|} - \sum_{j=1}^{l} \boldsymbol{\varphi}_{fb_{j}}^{\top} \mathbf{v}_{j} \boldsymbol{\xi}_{fb_{j}} \right).
	\end{aligned}
	\end{equation*}
	Therefore $P\left(l+1\right)$ is also true when $P\left(l\right)$ is true. Therefore, by the principle of induction, $P(l)$ holds $\forall \, l \in \left\{1,2,\hdots,k\right\}$, $\forall \, k\in \left\{1,2,\hdots,q\right\}$, completing the second part of the proof. \par
	
	Now, consider the bound on the transformed nonparametric uncertainty $\boldsymbol{\xi}_{fb_{k}}$ is given as follows.
	\begin{equation*}
		\left\| \boldsymbol{\xi}_{fb_{k}} \right\| \leq  \dfrac{ 2 \left(1+\delta_{2}\right)^{k-2}}{\delta_{1} \delta_{2}^{k-1}} \overline{\xi}_{y}, \quad k>1.
	\end{equation*}
	For $k=1$, $\left\| \boldsymbol{\xi}_{fb_{k}} \right\| \leq \dfrac{\overline{\xi}_{y}}{\delta_{1}}$. For $l=k$, $\boldsymbol{\xi}_{fb_{k}} = \boldsymbol{\xi}_{f_{l}} \left(t_{k}\right)$ thus,
	\begin{equation*}
		\left\| \boldsymbol{\xi}_{fb_{k+1}} \right\| \leq \dfrac{1}{\left\| \mathbf{v}_{k+1} \right\|} \left( \left\| \dfrac{\boldsymbol{\xi}_{f} \left(t_{k}\right)}{\boldsymbol{\varphi}_{f} \left(t_{k}\right)} \right\| + \sum_{j=1}^{k} \left\| \boldsymbol{\varphi}_{fb_{j}}^{\top} \right\| \left\| \mathbf{v}_{j} \right\| \left\| \boldsymbol{\xi}_{fb_{j}} \right\| \right).
	\end{equation*}
	Using \textit{Algorithm} $\ref{alg:identification}$ and $\eqref{eq:upper_bound_uncertainty}$, we obtain  the properties $\left\| \boldsymbol{\xi}_{f} \right\| \leq \overline{\xi}_{y}$, $\left\| \mathbf{v}_{1} \right\| = 1$, $\left\| \boldsymbol{\varphi}_{fb_{j}} \right\|=1$, $\delta_{2} < \left\| \mathbf{v}_{j} \right\| < 1$, $\delta_{2} < \left\| \mathbf{v}_{k+1} \right\| < 1$, and substituting bound on $\left\| \boldsymbol{\xi}_{fb_{k}} \right\|$, we obtain the following.
	\begin{equation*}
	\begin{aligned}
		\left\| \boldsymbol{\xi}_{fb_{k+1}} \right\| & \leq \dfrac{1}{\delta_{2}} \left( \dfrac{\overline{\xi}_{y}}{\delta_{1}} + \dfrac{\overline{\xi}_{y}}{\delta_{1}} + \sum_{j=2}^{k} \dfrac{ 2 \left(1+\delta_{2}\right)^{j-2}}{\delta_{1} \delta_{2}^{j-1}} \overline{\xi}_{y} \right), \\
			& \leq \dfrac{1}{\delta_{2}} \left( \dfrac{2}{\delta_{1}}\overline{\xi}_{y} + \dfrac{2}{\delta_{1} \delta_{2}} \sum_{j=2}^{k} \dfrac{\left(1+\delta_{2}\right)^{j-2}}{\delta_{2}^{j-2}} \overline{\xi}_{y} \right)
	\end{aligned}
	\end{equation*}
	Using $1+a+a^{2}+\hdots + a^{k-1} = \dfrac{a^{k}-1}{a-1}$.
	\begin{equation*}
		\left\| \boldsymbol{\xi}_{fb_{k+1}} \right\| \leq \dfrac{ 2 \overline{\xi}_{y}}{\delta_{1} \delta_{2}} \left( 1 + \dfrac{1}{\delta_{2}} \dfrac{\left(\dfrac{1+\delta_{2}}{\delta_{2}}\right)^{k-1}-1}{\dfrac{\left(1+\delta_{2}\right)}{\delta_{2}}-1} \right)
	\end{equation*}
	Further simplifying, we obtain the following.
	\begin{equation*}
		\left\| \boldsymbol{\xi}_{fb_{k+1}} \right\| \leq \dfrac{ 2 \overline{\xi}_{y}}{\delta_{1} \delta_{2}} \left( 1 + \dfrac{\left(1+\delta_{2}\right)^{k-1}-\delta_{2}^{k-1}}{\delta_{2}^{k-1}} \right)
	\end{equation*}
	On further simplification, we get.
	\begin{equation*}
		\left\| \boldsymbol{\xi}_{fb_{k+1}} \right\| \leq  \dfrac{2 \left(1+\delta_{2}\right)^{k-1}}{\delta_{1} \delta_{2}^{k}} \overline{\xi}_{y}, \quad k>1.
	\end{equation*}
Further, $\left\| \boldsymbol{\xi}_{fb} \right\| \leq \left\| \boldsymbol{\xi}_{fb_{1}} \right\| + \left\| \boldsymbol{\xi}_{fb_{2}} \right\| + \hdots + \left\| \boldsymbol{\xi}_{fb_{q}} \right\| \leq \dfrac{\overline{\xi}_{y}}{\delta_{1}} + \dfrac{2}{\delta_{1} \delta_{2}} \overline{\xi}_{y} + \dfrac{2 \left(1+\delta_{2}\right)}{\delta_{1} \delta_{2}^{2}} \overline{\xi}_{y} + \hdots + \dfrac{2 \left(1+\delta_{2}\right)^{q-2}}{\delta_{1} \delta_{2}^{q-1}} \overline{\xi}_{y}$, which simplifies to the following.
\begin{equation*}
	\left\| \boldsymbol{\xi}_{fb} \right\| \leq \dfrac{\overline{\xi}_{y}}{\delta_{1}} + \dfrac{2}{\delta_{1} \delta_{2}} \left( 1 + \dfrac{1+\delta_{2}}{\delta_{2}} + \hdots + \dfrac{\left(1+\delta_{2}\right)^{q-2}}{\delta_{2}^{q-2}} \right) \overline{\xi}_{y}.
\end{equation*}
Using the geometric series identity $1+a+a^{2}+\hdots + a^{k-1} = \dfrac{a^{k}-1}{a-1}$, we further simplify the bound.
\begin{equation*}
	\left\| \boldsymbol{\xi}_{fb} \right\| \leq \dfrac{\overline{\xi}_{y}}{\delta_{1}} + \dfrac{2}{\delta_{1} \delta_{2}} \left( \dfrac{\left(\dfrac{1+\delta_{2}}{\delta_{2}}\right)^{q-1}-1}{\dfrac{\left(1+\delta_{2}\right)}{\delta_{2}}-1} \right) \overline{\xi}_{y} = \dfrac{\overline{\xi}_{y}}{\delta_{1}} + \dfrac{2}{\delta_{1}} \dfrac{\left(1+\delta_{2}\right)^{q-1} - \delta_{2}^{q-1}}{\delta_{2}^{q-1}} \overline{\xi}_{y},
\end{equation*}
\begin{equation*}
	\left\| \boldsymbol{\xi}_{fb} \right\| \leq \left( \dfrac{1}{\delta_{1}} + \dfrac{2 \left( \left(1+\delta_{2}\right)^{q-1} - \delta_{2}^{q-1} \right)}{\delta_{1} \delta_{2}^{q-1}} \right) \overline{\xi}_{y}.
\end{equation*}
This completes the proof.
\end{proof}

\begin{remark}
The parameter $\delta_{1}$ sets the minimum acceptable norm of the filtered regressor vector $\boldsymbol{\varphi}_{f} \left( t \right)$ before the first orthonormal column is accepted into $\boldsymbol{\Phi}_{fb}$. It should be chosen in accordance with the system's operating region, specifically, proportional to the expected minimum magnitude of the regressor in normal operation. For instance, in a quadrotor hovering at a minimum altitude $h$ with a minimum propeller voltage $V$, a suitable choice is $\delta_{i} = \epsilon \left\| \begin{bmatrix} h & V \end{bmatrix} \right\|$ for some $\epsilon \in \left(0,1\right)$, where the norm is computed using normalized, dimensionless quantities to avoid unit inconsistency. A value of $\epsilon$ close to 1 accepts only strongly excited directions, which slows basis construction but improves robustness to noise. A value too close to 0 risks including nearly zero directions, which amplifies the effect of sensor noise and nonparametric uncertainty in the transformed system, as reflected in the upper bound $\overline{\xi}_{fb}$ in Theorem $\ref{theorem:algorithm1_lip}$.
\end{remark}

\begin{remark}
The parameter $\delta_{2}$ controls the minimum acceptable residual norm $\left\| \mathbf{v}_{k} \right\|$ after orthogonal projection, and therefore governs the degree of linear independence required between successive columns of $\boldsymbol{\Phi}_{fb}$. A natural and practically motivated choice is $\delta_{2} = \delta_{1} \sin(\theta_{\text{min}})$, where $\theta_{\text{min}}$ is the minimum acceptable angle between the new candidate direction and the existing subspace spanned by previously accepted columns. For example, choosing $\theta_{\text{min}}$ = $1^{\circ}$ gives $\delta_{2} = \delta_{1} \cos(1^{\circ})$, which accepts a new column only if it is sufficiently linearly independent of all previously accepted columns. This threshold is consistent with the least count of typical sensor measurements and avoids numerical issues arising from near-linear dependence between regressor samples. Selecting a larger $\theta_{\text{min}}$ (equivalently, larger $\delta_{2}$) enforces stricter linear independence, yielding a better-conditioned basis at the cost of slower column accumulation. This introduces a trade-off: a larger $\delta_{2}$ produces a better-conditioned $\boldsymbol{\Psi}_{fb}$ and a tighter ultimate bound $\overline{\xi}_{fb}$, but may delay the onset of combined adaptation by increasing the time $t_{q}$ required to complete the orthogonal basis.
\label{remark:delta_2}
\end{remark}

\begin{remark}
	The bound derived in Theorem \ref{theorem:algorithm1_lip} is conservative. In the expression for the upper bound on $\left\| \boldsymbol{\xi}_{fb_{k+1}} \right\|$, the vector $\mathbf{v}_{j}$ satisfies $\delta_{2} < \left\| \mathbf{v}_{j} \right\| < 1$. However, while deriving the bound, we conservatively assumed $\left\| \mathbf{v}_{j} \right\| < 1$, which results in an overestimation. In practice, however, the actual uncertainty magnitude is often significantly lower, as supported by numerical simulations.
\end{remark}

We next propose an update law for estimating the system parameters using the outputs of Algorithm \ref{alg:identification}, and establish several stability properties essential for closed-loop stability analysis.

\begin{theorem}
	Let $\hat{\mathbf{w}}$ denote the estimate of the unknown constant vector $\mathbf{w}$, and the parameter estimation error denoted by $\tilde{\mathbf{w}}$ be defined as $\tilde{\mathbf{w}} = \hat{\mathbf{w}} - \mathbf{w}$. The proposed update law for online estimate $\hat{\mathbf{w}}$ is as follows.
	\begin{equation}
		\dot{\hat{\mathbf{w}}} = \gamma_{w} \boldsymbol{\Gamma}_{w} \boldsymbol{\Psi}_{fb} \left(\mathbf{y}_{fb} - \boldsymbol{\Psi}_{fb}^{\top} \hat{\mathbf{w}} \right),
		\label{eq:adaptation_law_system_parameter}
	\end{equation}	
	where $\gamma_{w}=0$ when $k\leq q$ and $\gamma_{w}=1$ when $k>q$ (\textit{Algorithm} $\ref{alg:identification}$). $\boldsymbol{\Gamma}_{w}$ is a positive definite diagonal gain matrix. Substituting $\mathbf{y}_{fb} = \boldsymbol{\Psi}_{fb}^{\top} \mathbf{w} + \boldsymbol{\xi}_{fb}$ from Theorem $\ref{theorem:algorithm1_lip}$ into $\eqref{eq:adaptation_law_system_parameter}$ generates the following parameter estimation error dynamics
	\begin{equation}
		\dot{\tilde{\mathbf{w}}} = \gamma_{w} \boldsymbol{\Gamma}_{w} \boldsymbol{\Psi}_{fb} \left( -\boldsymbol{\Psi}_{fb}^{\top} \tilde{\mathbf{w}} + \boldsymbol{\xi}_{fb} \right).
		\label{eq:parameter_estimation_error_dynamics}
	\end{equation}
	The following statements hold with the designed update law $\eqref{eq:adaptation_law_system_parameter}$.
	 \begin{enumerate}
	 	\item Let $t_{q}$ denote the time instant at which the matrix $\boldsymbol{\Phi}_{fb}$ gets full rank. Then, $\hat{\mathbf{w}}\left(t\right) = \hat{\mathbf{w}} \left(t_{0}\right)$, $\forall t_{0} \leq t \leq t_{q}$.
	 	\item The parameter estimation error is uniformly ultimately bounded $\forall \, t\geq t_{0}$.
	 	\begin{equation}
	 		\left\| \tilde{\mathbf{w}} \left(t\right) \right\|^{2} \leq \begin{cases}
	 			\left\| \hat{\mathbf{w}} \left(t_{0}\right) - \mathbf{w} \right\|^{2}, & t_{0} \leq t \leq t_{q} \\
	 			\dfrac{\lambda_{\text{max}}\left(\boldsymbol{\Gamma}_{w}^{-1}\right)}{\lambda_{\text{min}} \left(\boldsymbol{\Gamma}_{w}^{-1}\right)}  \left\| \tilde{\mathbf{w}} \left(t_{0}\right) \right\|^{2} e^{-k_{w} \left(t-t_{q}\right)} + \dfrac{\lambda_{\text{max}}\left(\boldsymbol{\Gamma}_{w}^{-1}\right)}{\lambda_{\text{min}} \left(\boldsymbol{\Gamma}_{w}^{-1}\right)} \left(1-e^{-k_{w} \left(t-t_{q}\right)}\right) \overline{\xi}_{fb}^{2}, & t>t_{q}
	 		\end{cases}.
	 		\label{eq:parameter_estimation_bound}
	 	\end{equation}
	 	where $k_{w} = \lambda_{\min}\left(\boldsymbol{\Gamma}_{w} \right)$.
	 	\item In the absence of a nonparametric uncertainty, that is, $\boldsymbol{\xi}=0$, the parameter estimation error dynamics is
	 	\begin{equation}
	 		\dot{\tilde{\mathbf{w}}} = - \boldsymbol{\Gamma}_{w} \tilde{\mathbf{w}}, \quad \forall t>t_{q}.
	 		\label{eq:parameter_estimation_error_dynamics_monotonic}
	 	\end{equation}
	 \end{enumerate}
	 \label{theorem:parameter_estimation_bounds}
\end{theorem}

\begin{proof}
Consider the following Lyapunov function candidate in terms of the parameter estimation error $\tilde{\mathbf{w}}$.
\begin{equation}
	V_{w} = \dfrac{1}{2} \tilde{\mathbf{w}}^{\top} \boldsymbol{\Gamma}_{w}^{-1} \tilde{\mathbf{w}}.
	\label{eq:V_estimation}
\end{equation}
The time derivative of $V_{w}$ along the error trajectory $\eqref{eq:parameter_estimation_error_dynamics}$ is
\begin{equation*}
	\dot{V}_{w} = - \gamma_{w} \tilde{\mathbf{w}}^{\top} \boldsymbol{\Psi}_{fb} \boldsymbol{\Psi}_{fb}^{\top} \tilde{\mathbf{w}} + \gamma_{w} \tilde{\mathbf{w}}^{\top} \boldsymbol{\Psi}_{fb} \boldsymbol{\xi}_{fb}.
\end{equation*}
As defined in the update law $\eqref{eq:adaptation_law_system_parameter}$, $\gamma_{w}$ goes from $0$ to $1$ once $k>q$, that is, when sufficient data has been accumulated. Therefore, for the initial time duration $t_{0} \leq t \leq t_{q}$, we have $\boldsymbol{\gamma}_{w}=0$, implying that $\dot{V}_{w}=0$, and hence $\tilde{\mathbf{w}}\left(t\right) = \tilde{\mathbf{w}} \left(t_{0}\right)$, and  $\hat{\mathbf{w}}\left(t\right) = \hat{\mathbf{w}} \left(t_{0}\right)$ over this interval. Once $k>q$, that is, $t>t_{q}$, we have $\boldsymbol{\gamma}_{w} = 1$. Since $k>q$ implies successful creation of an orthogonal matrix $\boldsymbol{\Psi}_{fb}$, inferring $\boldsymbol{\Psi}_{fb}\boldsymbol{\Psi}_{fb}^{\top} = \mathbf{I}$ and $\left\| \boldsymbol{\Psi}_{fb} \boldsymbol{\xi}_{fb} \right\| = \left\| \boldsymbol{\xi}_{fb} \right\|$. Substituting this into $\dot{V}_{w}$, we obtain the following.
\begin{equation*}
	\dot{V}_{w} \leq - \left\| \tilde{\mathbf{w}} \right\|^{2} + \left\| \tilde{\mathbf{w}} \right\| \left\| \boldsymbol{\xi}_{fb} \right\|, \quad t>t_{q},
\end{equation*}
using $2ab \leq a^{2}+b^{2}$ for any $a,b \in \mathbb{R}$, further simplifies as follows.
\begin{equation}
	\dot{V}_{w} \leq - \dfrac{1}{2} \left\| \tilde{\mathbf{w}} \right\|^{2} + \dfrac{1}{2} \left\| \boldsymbol{\xi}_{fb} \right\|^{2}, \quad t>t_{q}.
	\label{eq:Vw_dot}
\end{equation}
From \textit{Theorem} $\ref{theorem:algorithm1_lip}$, $\left\| \boldsymbol{\xi}_{fb} \right\| \leq \overline{\xi}_{fb}$, leading to $\dot{V}_{w} \leq - \dfrac{1}{2}\left\| \tilde{\mathbf{w}} \right\|^{2} + \dfrac{1}{2} \overline{\xi}_{fb}^{2}$. The Lyapunov function $V_{w}$ satisfies $0.5\lambda_{\text{min}} \left(\boldsymbol{\Gamma}_{w}^{-1}\right) \left\| \tilde{\mathbf{w}} \right\|^{2} \leq V_{w} \leq 0.5\lambda_{\text{max}} \left(\boldsymbol{\Gamma}_{w}^{-1}\right) \left\| \tilde{\mathbf{w}} \right\|^{2}$, we can express $\dot{V}_{w}$ as below.
	\begin{equation*}
		\dot{V}_{w} \leq - k_{w} V_{w} + \rho_{w}, \quad k_{w}=\dfrac{1}{\lambda_{\text{max}} \left(\boldsymbol{\Gamma}_{w}^{-1}\right)}, \quad \rho_{w} = \dfrac{1}{2} \overline{\xi}_{fb}^{2}.
	\end{equation*}
	Using the comparison lemma (Lemma A.3.2, \cite{farrell2006adaptive}), the solution to the above differential inequality is
	\begin{equation*}
		V_{w}\left(t\right) \leq \left(V_{w}\left(t_{0}\right) - \dfrac{\rho_{w}}{k_{w}}\right) e^{-k_{w} \left(t-t_{q}\right)}+\dfrac{\rho_{w}}{k_{w}}, \quad \forall \, t > t_{q},
	\end{equation*}
	which is further expressed in terms of the parameter estimation error $\tilde{\mathbf{w}}$.
	\begin{equation*}
		\left\| \tilde{\mathbf{w}} \left(t\right) \right\|^{2} \leq \dfrac{1}{0.5\lambda_{\text{min}} \left(\boldsymbol{\Gamma}_{w}^{-1}\right)} \left(V_{w}\left(t_{0}\right) - \dfrac{\rho_{w}}{k_{w}}\right) e^{-k_{w} \left(t-t_{q}\right)}+ \dfrac{1}{0.5\lambda_{\text{min}} \left(\boldsymbol{\Gamma}_{w}^{-1}\right)} \dfrac{\rho_{w}}{k_{w}},
	\end{equation*}
	substituting $V_{w}\left(t_{0}\right) \leq 0.5\lambda_{\text{max}}\left(\boldsymbol{\Gamma}_{w}^{-1} \right) \left\| \tilde{\mathbf{w}} \left(t_{0}\right) \right\|^{2}$, we obtain
	\begin{equation*}
		\left\| \tilde{\mathbf{w}} \left(t\right) \right\|^{2} \leq \dfrac{1}{0.5\lambda_{\text{min}} \left(\boldsymbol{\Gamma}_{w}^{-1}\right)} \left(0.5\lambda_{\text{max}}\left(\boldsymbol{\Gamma}_{w}^{-1}\right) \left\| \tilde{\mathbf{w}} \left(t_{0}\right) \right\|^{2} - 0.5\lambda_{\text{max}}\left(\boldsymbol{\Gamma}_{w}^{-1}\right)\overline{\xi}_{fb}^{2} \right) e^{-k_{w} \left(t-t_{q}\right)} + \dfrac{\lambda_{\text{max}}\left(\boldsymbol{\Gamma}_{w}^{-1}\right)}{\lambda_{\text{min}} \left(\boldsymbol{\Gamma}_{w}^{-1}\right)} \overline{\xi}_{fb}^{2} , \quad \forall \, t > t_{q}.
	\end{equation*}
	Finally, in the special case where $\boldsymbol{\xi}=0$, we have $\boldsymbol{\xi}_{fb}=0$, and  $\boldsymbol{\Psi}_{fb} \boldsymbol{\Psi}_{fb}^{\top} = \mathbf{I}$, recovering the monotonic parameter estimation error dynamics $\eqref{eq:parameter_estimation_error_dynamics_monotonic}$,	thus completing the proof.
\end{proof}

\begin{remark}
	The parameter estimation error dynamics in \eqref{eq:parameter_estimation_error_dynamics_monotonic} is driven by past system data. Notably, the individual parameter error components are decoupled, enabling independent gain tuning for faster convergence. The convergence rate $k_w$ depends solely on the user-defined gain $\boldsymbol{\Gamma}_{w}$, unlike DREM-based methods \cite{aranovskiy2016parameters}, which involve determinant terms in the update law.
\end{remark}

We have proved that the parameter estimate $\hat{\mathbf{w}}$ is uniformly ultimately bounded. However, the controller gain estimates, which are deduced using the matching condition and parameter estimate $\hat{\mathbf{w}}$, require $\hat{\mathbf{A}}$ and $\hat{\boldsymbol{\Lambda}}$. The following lemma proves the boundedness of $\hat{\mathbf{A}}$ and $\hat{\boldsymbol{\Lambda}}$ in terms of $\hat{\mathbf{w}}$.

\begin{proposition}
	The following inequalities hold
	\begin{equation}
		\left\| \mathbf{B}^{\dagger} \tilde{\mathbf{A}} \right\|_{F} \leq \left\| \tilde{\mathbf{w}} \right\| \text{ and } \left\| \tilde{\boldsymbol{\Lambda}} \right\|_{F} \leq \left\| \tilde{\mathbf{w}} \right\|.
		\label{eq:parameter_estimation_error_bound}
	\end{equation}
	\label{prop:parameter_estimation_error_bound}
\end{proposition}
\begin{proof} The inequalities in $\eqref{eq:parameter_estimation_error_bound}$ follow directly from the fact that $\tilde{\mathbf{w}} = vec\left(\begin{bmatrix} \mathbf{B}^{\dagger} \tilde{\mathbf{A}} & \tilde{ \boldsymbol{\Lambda}} & \hat{ \boldsymbol{\Lambda}} \tilde{ \boldsymbol{\Theta}}^{\top} + \tilde{ \boldsymbol{\Lambda}}  \boldsymbol{\Theta}^{\top} \end{bmatrix}^{\top}\right)$, so that $\left\| \tilde{\mathbf{w}} \right\|^{2} = \left\| \mathbf{B}^{\dagger} \tilde{\mathbf{A}} \right\|_{F}^{2} + \left\| \tilde{ \boldsymbol{\Lambda}} \right\|_{F}^{2} + \left\| \hat{ \boldsymbol{\Lambda}} \tilde{ \boldsymbol{\Theta}}^{\top} + \tilde{ \boldsymbol{\Lambda}}  \boldsymbol{\Theta}^{\top} \right\|_{F}^{2}$, and each component norm is therefore bounded by the total norm $\left\| \tilde{\mathbf{w}} \right\|$. (See Appendix B for the complete derivation.)
\end{proof}

Having established that the parameter estimate $\hat{\mathbf{w}}$ is uniformly ultimately bounded, we now turn to estimating the controller gains using the system parameter estimates $\hat{\mathbf{A}}$ and $\hat{\boldsymbol{\Lambda}}$, obtained through the matching condition. Since nonparametric uncertainties are considered, the parameter estimates obtained online may not be sufficiently accurate, though they remain bounded. Specifically, the sign structure and magnitude of the diagonal entries of $\hat{\boldsymbol{\Lambda}}$ must be consistent with those of $\boldsymbol{\Lambda}$ for the controller gain updates to be meaningful. If the diagonal entries of $\hat{\boldsymbol{\Lambda}}$ have opposite signs compared to those of $\boldsymbol{\Lambda}$, the resulting controller gain updates may increase control effort in undesirable directions, potentially destabilizing the closed-loop system. Similarly, if the diagonal entries of $\hat{\boldsymbol{\Lambda}}$ are too small in magnitude, the computed gain updates may demand excessive control authority, leading to actuator saturation. To avoid such adverse scenarios, the controller gain update is conditioned on the following inequalities being satisfied.
\begin{equation}
	\hat{\boldsymbol{\Lambda}} \boldsymbol{\Lambda}_{s}>0, \quad \text{ and } \quad \lambda_{\text{min}} \left( \hat{\boldsymbol{\Lambda}} \right) > \underline{\lambda},
	\label{eq:condition_Lambda}
\end{equation}
where $\underline{\lambda}$ is the known positive lower bound on the minimum eigenvalue of $\boldsymbol{\Lambda}$ [see \textit{Assumption} $\ref{assump:Lambda}$]. The first condition ensures that the sign structure of $\hat{\boldsymbol{\Lambda}}$ is consistent with that of $\boldsymbol{\Lambda}$ using the known sign matrix $\boldsymbol{\Lambda}_{s}$ from Assumption $\ref{assump:Lambda}$. The second condition ensures that $\hat{\boldsymbol{\Lambda}}$ is sufficiently well-conditioned for use in the gain computation. When either condition is violated, the memory term does not update the parameter estimate, and the adaptation reverts to the standard $\sigma$-modification update rule, for which boundedness of all closed-loop signals is guaranteed by Proposition $\ref{proposition:direct_adaptive}$. These conditions are incorporated into the combined adaptation law in the following subsection through the switching parameter $\gamma_{i}$.

In the presence of nonparametric uncertainties, the parameter estimates obtained from a single dataset may not converge precisely to their ideal values, resulting in a non-zero ultimate bound. Collecting multiple datasets and incorporating them into the parameter update law can mitigate the effect of nonparametric uncertainties by averaging their influence across datasets, thereby reducing the effective noise level in the estimation error dynamics. 

\subsection{Combined Adaptation in Finite Excitation Condition}
\label{subsec:combined_adaptation_in_finite_excitation_condition}
The adaptation laws to update the estimates of the controller gains using the system parameter estimates are discussed next. Let $\hat{\mathbf{A}}$, $\hat{\boldsymbol{\Lambda}}$, and $\hat{\boldsymbol{\Lambda}}_{\boldsymbol{\Theta}}$ denote the estimates of $\mathbf{A}$, $\boldsymbol{\Lambda}$, and $\boldsymbol{\Lambda} \boldsymbol{\Theta}^{\top}$, respectively, obtained using the update law $\eqref{eq:adaptation_law_system_parameter}$. The indirect estimates of the controller gains $\mathbf{K}_{x}$ and $\mathbf{K}_{r}$ can then be computed from $\eqref{eq:matching_condition_ideal_gains}$. Notably, the accuracy of the indirect estimates $\mathbf{K}_{x}$ and $\mathbf{K}_{r}$ depends on the quality of $\hat{\boldsymbol{\Lambda}}$. Therefore, we incorporate the inequalities condition $\eqref{eq:condition_Lambda}$ to modify the adaptation laws in $\eqref{eq:controller_gain_update_law_gradient}$.

\begin{equation}
	\begin{aligned}
		\dot{\hat{\mathbf{K}}}_{x} =& - \left(1-\gamma_{i}\right) \hat{\mathbf{K}}_{x} - \mathbf{x} \mathbf{e}^{\top} \mathbf{P} \mathbf{B} \boldsymbol{\Lambda}_{s} + \gamma_{i} \left\{ \left( \mathbf{B}^{\dagger} \mathbf{A}_{r} - \mathbf{B}^{\dagger} \hat{\mathbf{A}} \right)^{\top} - \hat{\mathbf{K}}_{x} \hat{\boldsymbol{\Lambda}} \right\} \boldsymbol{\Lambda}_{s}, \\
		\dot{\hat{\mathbf{K}}}_{r} =& - \left(1-\gamma_{i}\right) \hat{\mathbf{K}}_{r} - \mathbf{r} \mathbf{e}^{\top} \mathbf{P} \mathbf{B} \boldsymbol{\Lambda}_{s} + \gamma_{i} \left\{ \left( \mathbf{B}^{\dagger} \mathbf{B}_{r} \right)^{\top} - \hat{\mathbf{K}}_{r} \hat{\boldsymbol{\Lambda}} \right\} \boldsymbol{\Lambda}_{s}, \\
		\dot{\hat{\boldsymbol{\Theta}}} =& - \left(1-\gamma_{i}\right) \hat{\boldsymbol{\Theta}} + \boldsymbol{\phi}\left( \mathbf{x} \right) \mathbf{e}^{\top} \mathbf{P} \mathbf{B} \boldsymbol{\Lambda}_{s} + \gamma_{i} \left( \hat{\boldsymbol{\Lambda}}_{\boldsymbol{\Theta}}^{\top} - \hat{\boldsymbol{\Theta}} \hat{\boldsymbol{\Lambda}} \right) \boldsymbol{\Lambda}_{s},
	\end{aligned}
	\label{eq:controller_gain_update_law_proposed}
\end{equation}
where $\gamma_{i}$ is defined as
\begin{equation}
	\gamma_{i} = \begin{cases} \gamma_{w}, & \text{ if } \hat{\boldsymbol{\Lambda}} \boldsymbol{\Lambda}_{s}>0 \text{ and } \lambda_{\text{min}} \left( \hat{\boldsymbol{\Lambda}} \right) > \underline{\lambda} \\ 0, & \text{ otherwise}.	\end{cases}
	\label{eq:gamma_i}
\end{equation}
Here, $\gamma_{w} \in \left\{ 0,1 \right\}$ is the excitation flag defined in Theorem $\ref{theorem:parameter_estimation_bounds}$, which equals $1$ only when the orthogonal basis is complete (i.e., $k > q$).

\begin{proposition}
	The proposed adaptation laws $\eqref{eq:controller_gain_update_law_proposed}$, which update the estimates of the controller gains $\hat{\mathbf{K}}_{x}$, $\hat{\mathbf{K}}_{r}$, and $\hat{\boldsymbol{\Theta}}$, results in the following estimation error dynamics.
	\begin{equation}
	\begin{aligned}
		\dot{\tilde{\mathbf{K}}}_{x} &= - \left\{ \left(1-\gamma_{i}\right) \hat{\mathbf{K}}_{x} + \gamma_{i} \tilde{\mathbf{K}}_{x} \hat{\boldsymbol{\Lambda}} \boldsymbol{\Lambda}_{s} \right\} - \mathbf{x} \mathbf{e}^{\top} \mathbf{P} \mathbf{B} \boldsymbol{\Lambda}_{s} - \gamma_{i} \left( \mathbf{B}^{\dagger} \tilde{\mathbf{A}} + \tilde{\boldsymbol{\Lambda}} \mathbf{K}_{x}^{\top} \right)^{\top} \boldsymbol{\Lambda}_{s}, \\
		\dot{\tilde{\mathbf{K}}}_{r} &= - \left\{ \left(1-\gamma_{i}\right) \hat{\mathbf{K}}_{r} + \gamma_{i} \tilde{\mathbf{K}}_{r} \hat{\boldsymbol{\Lambda}} \boldsymbol{\Lambda}_{s} \right\} - \mathbf{r} \mathbf{e}^{\top} \mathbf{P} \mathbf{B} \boldsymbol{\Lambda}_{s} - \gamma_{i} \mathbf{K}_{r} \tilde{\boldsymbol{\Lambda}} \boldsymbol{\Lambda}_{s}, \\
		\dot{\tilde{\boldsymbol{\Theta}}} &= - \left\{ \left(1-\gamma_{i}\right) \hat{\boldsymbol{\Theta}} + \gamma_{i} \tilde{\boldsymbol{\Theta}} \hat{\boldsymbol{\Lambda}} \boldsymbol{\Lambda}_{s} \right\} + \boldsymbol{\phi} \left(\mathbf{x}\right) \mathbf{e}^{\top} \mathbf{P} \mathbf{B} \boldsymbol{\Lambda}_{s} - \gamma_{i} \boldsymbol{\Theta} \tilde{\boldsymbol{\Lambda}} \boldsymbol{\Lambda}_{s}.
	\end{aligned}
	\label{eq:controller_gain_estimation_error_dynamics}
	\end{equation}
	\label{prop:controller_gain_estimation_error_dynamics}
\end{proposition}

\begin{proof} Refer Appendix C. \end{proof}

\begin{remark}
	The proposed adaptation laws $\eqref{eq:controller_gain_update_law_proposed}$ reduce to the conventional adaptation laws with $\sigma$-modification, as given in \eqref{eq:controller_gain_update_law_gradient}, when $\gamma_{i}=0$. The system parameter estimation errors $\tilde{\mathbf{A}}$ and $\tilde{\boldsymbol{\Lambda}}$, which appear in the controller gain estimation error dynamics $\eqref{eq:controller_gain_estimation_error_dynamics}$, have already been shown to be bounded [see \textit{Proposition} $\ref{prop:parameter_estimation_error_bound}$]. The uniform ultimate boundedness of the closed-loop system is established in the following theorem. Unity adaptation gains are used in the adaptation laws $\eqref{eq:controller_gain_update_law_proposed}$, and retained throughout the subsequent analysis.
\end{remark}

\begin{theorem}
	Consider the class of uncertain nonlinear systems described by  $\eqref{eq:system_dynamics}$, under Assumptions \ref{assump:controllable}, \ref{assump:Lambda}, \ref{assump:matching_condition}, and \ref{assump:fe_condition}, driven by the control and adaptation laws given in $\eqref{eq:control_input_modified}$ and  $\eqref{eq:controller_gain_update_law_proposed}$, respectively. The closed-loop system shown in Figure $\ref{fig:block_diagram}$ exhibits the following properties.
	\begin{enumerate}
		\item The closed-loop signals are uniformly ultimately bounded $\forall \, t \geq t_{0}$.
		\item The combined error vector, denoted by $\boldsymbol{\chi} = \begin{bmatrix} \mathbf{e}^{\top} & vec^{\top}\left(\tilde{\mathbf{K}}_{x}\right) & vec^{\top}\left(\tilde{\mathbf{K}}_{r}\right) & vec^{\top}\left(\tilde{ \boldsymbol{\Theta}}\right) & \tilde{\mathbf{w}}^{\top} \end{bmatrix}^{\top}$, satisfies the following inequality.
		\begin{equation*}
			\left\| \boldsymbol{\chi} \left(t\right) \right\|^{2} \leq \kappa_{1} \left\| \boldsymbol{\chi} \left(t_{q}\right) \right\|^{2} e^{-k\left(t-t_{q}\right)} + \rho_{ub}, \quad \forall \, t>t_{q},
		\end{equation*}
		where $\kappa_{1}>0$, $\rho_{ub}>0$ denotes the ultimate bound, $k>0$ is the exponential decay rate, and $t_{q}$ is the time instant at which the matrix $\boldsymbol{\Phi}_{fb}$ becomes orthogonal and the condition given in $\eqref{eq:condition_Lambda}$ is also satisfied.
		\item The exponential decay rate $k$ is independent of the excitation level of the regressor signal $\boldsymbol{\varphi}\left(\mathbf{x}, \mathbf{u}\right)$, and the ultimate bound $\rho_{ub}$ depends on nonparametric uncertainty bound $\overline{\xi}$, the parameter estimation error $\tilde{\mathbf{w}}$, and the ideal controller gains and system parameters.
	\end{enumerate}
\end{theorem}

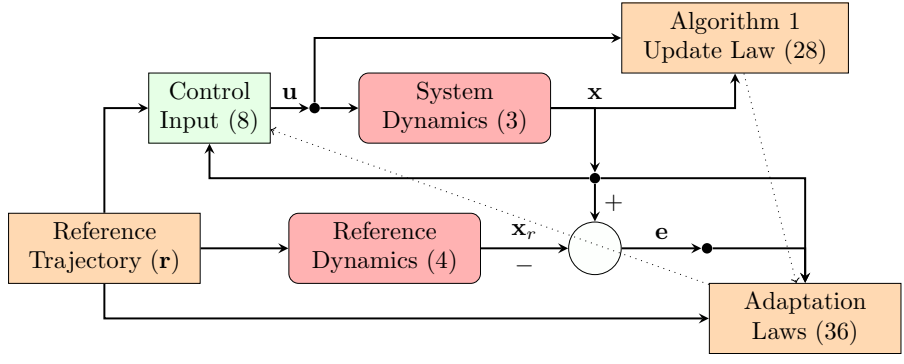
\begin{figure}
	\centering
	\resizebox{0.7\linewidth}{!}{
		\begin{tikzpicture}[node distance=2cm]
			\node (reference) [process, text width=2.5cm] {Reference Trajectory $\left(\mathbf{r}\right)$ };
			\node (referencedynamics) [startstop, right of=reference, xshift=2cm, text width=2.5cm] {Reference Dynamics $\eqref{eq:reference_system}$};
			\node (add) [summer, right of=referencedynamics, xshift=1cm]{};
			\node (route1) [route, above of=add, yshift=-1cm]{};
			\node (route3) [route, right of=add, xshift=-0.4cm]{};
			\node (control) [io, above of=referencedynamics, xshift=-2.5cm, text width=1.5cm] {Control Input $\eqref{eq:control_input_modified}$};
			\node (route2) [route, right of=control, xshift=-0.5cm]{};
			\node (systemdynamics) [startstop, right of=control, xshift=1.5cm, text width=2.5cm] {System Dynamics $\eqref{eq:system_dynamics}$};
			\node (algorithm) [process, right of=systemdynamics, xshift=2cm, yshift=1cm, text width=3cm] {Algorithm $1$ Update Law $\eqref{eq:adaptation_law_system_parameter}$};
			\node (adaptationlaws) [process, right of=referencedynamics, xshift=4cm, yshift=-1cm, text width=2.5cm] {Adaptation Laws $\eqref{eq:controller_gain_update_law_proposed}$};
			
			\draw [arrow] (reference) |- (control);
			\draw [arrow] (control) -- node[anchor=south] {$\mathbf{u}$} (route2);
			\draw [arrow] (route2) -- (systemdynamics);
			\draw [arrow] (route2) |- (algorithm);
			\draw [arrow] (systemdynamics) -| (algorithm);
			\draw [arrow] (reference) -- (referencedynamics);
			\draw [arrow] (referencedynamics) -- node[anchor=south] {$\mathbf{x}_{r}$} node[anchor=north] {$-$} (add);
			\draw [arrow] (systemdynamics) -| node[anchor=south] {$\mathbf{x}$} (route1);
			\draw [arrow] (route1) -- node[anchor=west] {$+$} (add);
			\draw [arrow] (route1) -| (control);
			\draw [arrow] (reference) |- (adaptationlaws);
			\draw [arrow] (add) -- node [anchor=south] {$\mathbf{e}$} (route3);
			\draw [arrow] (route3) -| (adaptationlaws);
			\draw [arrow] (route1) -| (adaptationlaws);
			\draw [dotted,->] (algorithm) -- (adaptationlaws);
			\draw [dotted,->] (adaptationlaws) -- (control);
		\end{tikzpicture}
	}
	\caption{Block diagram of the closed-loop system.}
	\label{fig:block_diagram}
\end{figure}

\begin{proof}
The stability properties of the closed-loop system are established in two segments. The first segment deals with the initial time window, during which the orthogonal matrix $\boldsymbol{\Phi}_{fb}$ is being constructed. The second segment begins once an orthogonal matrix $\boldsymbol{\Phi}_{fb}$ is created successfully.  Consider the following Lyapunov function, formed as the sum of the Lyapunov functions defined in $\eqref{eq:V_gradient}$ and  $\eqref{eq:V_estimation}$.
\begin{equation}
	V = \mathbf{e}^{\top} \mathbf{P} \mathbf{e} + \tr(\tilde{\mathbf{K}}_{x}^{\top} \tilde{\mathbf{K}}_{x}  \boldsymbol{\Lambda}_{s}  \boldsymbol{\Lambda}) + \tr( \tilde{\mathbf{K}}_{r}^{\top} \tilde{\mathbf{K}}_{r}  \boldsymbol{\Lambda}_{s}  \boldsymbol{\Lambda}) + \tr(\tilde{ \boldsymbol{\Theta}}^{\top} \tilde{ \boldsymbol{\Theta}}  \boldsymbol{\Lambda}_{s}  \boldsymbol{\Lambda}) + \tilde{\mathbf{w}}^{\top} \tilde{\mathbf{w}},
	\label{eq:V_combined}
\end{equation}
where, for simplicity of presentation, $\boldsymbol{\Gamma}_{w}$ is set as the identity matrix. The time derivative of the Lyapunov function $V$ along the error trajectory defined in $\eqref{eq:tracking_error}$ is
\begin{equation*}
	\begin{aligned}
		\dot{V} =& \mathbf{e}^{\top} \left( \mathbf{A}_{r}^{\top} \mathbf{P} + \mathbf{P} \mathbf{A}_{r} \right) \mathbf{e} + 2 \mathbf{e}^{\top} \mathbf{P} \mathbf{B}  \boldsymbol{\Lambda} \tilde{\mathbf{K}}_{x}^{\top} \mathbf{x} + 2 \mathbf{e}^{\top} \mathbf{P} \mathbf{B}  \boldsymbol{\Lambda} \tilde{\mathbf{K}}_{r}^{\top} \mathbf{r} - 2 \mathbf{e}^{\top} \mathbf{P} \mathbf{B}  \boldsymbol{\Lambda} \tilde{ \boldsymbol{\Theta}}^{\top} \boldsymbol{\phi}(\mathbf{x}) \\
		& + 2 \mathbf{e}^{\top} \mathbf{P} \boldsymbol{\xi} + 2 \tr(\tilde{\mathbf{K}}_{x}^{\top} \dot{\tilde{\mathbf{K}}}_{x}  \boldsymbol{\Lambda}_{s}  \boldsymbol{\Lambda}) + 2 \tr(\tilde{\mathbf{K}}_{r}^{\top} \dot{\tilde{\mathbf{K}}}_{r}  \boldsymbol{\Lambda}_{s}  \boldsymbol{\Lambda}) + 2 \tr(\tilde{ \boldsymbol{\Theta}}^{\top} \dot{\tilde{ \boldsymbol{\Theta}}}  \boldsymbol{\Lambda}_{s}  \boldsymbol{\Lambda}) + 2 \tilde{\mathbf{w}}^{\top} \dot{\tilde{\mathbf{w}}}.
	\end{aligned}
\end{equation*}

In the first segment, that is, for $t_{0} \leq t \leq t_{q}$, we have $k\leq q$, which implies $\gamma_{w}=0$, which further implies $\gamma_{i}=0$ from $\eqref{eq:gamma_i}$. Therefore, retaining the stability results enumerated in \textit{Proposition} $\ref{proposition:direct_adaptive}$. \par

In the second segment, that is, for $t>t_{q}$, we have $k>q$, which implies $\gamma_{w}=1$. The parameter estimate $\hat{\mathbf{w}}$ is obtained using the adaptation law $\eqref{eq:adaptation_law_system_parameter}$, and $\hat{\boldsymbol{\Lambda}}$ is computed subsequently. If the condition given in $\eqref{eq:condition_Lambda}$ is not satisfied, then $\gamma_{i}=0$. Thus, the stability results from \textit{Proposition} $\ref{proposition:direct_adaptive}$ continue to hold.

When the condition given in $\eqref{eq:condition_Lambda}$ is satisfied, $\gamma_{i}=1$, and, we obtain the following expression for $\dot{V}$.
\begin{equation}
	\begin{aligned}
		\dot{V} &= -\mathbf{e}^{\top} \mathbf{Q} \mathbf{e} + 2\mathbf{e}^{\top} \mathbf{P} \boldsymbol{\xi} - 2 \tr(\tilde{\mathbf{K}}_{x}^{\top} \tilde{\mathbf{K}}_{x} \hat{ \boldsymbol{\Lambda}}  \boldsymbol{\Lambda}) - 2 \tr(\tilde{\mathbf{K}}_{r}^{\top} \tilde{\mathbf{K}}_{r} \hat{ \boldsymbol{\Lambda}}  \boldsymbol{\Lambda}) - 2 \tr(\tilde{ \boldsymbol{\Theta}}^{\top} \tilde{ \boldsymbol{\Theta}} \hat{ \boldsymbol{\Lambda}}  \boldsymbol{\Lambda}) \\
		& \quad - 2 \tr(\tilde{\mathbf{K}}_{x}^{\top} \left( \mathbf{B}^{\dagger} \tilde{\mathbf{A}} + \tilde{ \boldsymbol{\Lambda}} \mathbf{K}_{x}^{\top} \right)^{\top}  \boldsymbol{\Lambda}) - 2 \tr(\tilde{\mathbf{K}}_{r}^{\top} \mathbf{K}_{r} \tilde{ \boldsymbol{\Lambda}}  \boldsymbol{\Lambda}) - 2 \tr(\tilde{ \boldsymbol{\Theta}}^{\top}  \boldsymbol{\Theta} \tilde{ \boldsymbol{\Lambda}}  \boldsymbol{\Lambda}) + 2 \tilde{\mathbf{w}}^{\top} \dot{\tilde{\mathbf{w}}}.
	\end{aligned}
	\label{eq:V_dot_step1}
\end{equation}

Substituting $\dot{V}_{w}$ from $\eqref{eq:Vw_dot}$ in $\dot{V}$, and noting that the conditions $\hat{ \boldsymbol{\Lambda}}  \boldsymbol{\Lambda}_{s}>0$ and $ \boldsymbol{\Lambda}  \boldsymbol{\Lambda}_{s}>0$ imply $\hat{ \boldsymbol{\Lambda}}  \boldsymbol{\Lambda}>0$. Thus, the following inequality holds.
\begin{equation*}
	\begin{aligned}
		\dot{V} &\leq - \lambda_{\text{min}}\left(\mathbf{Q}\right) \left\| \mathbf{e} \right\|^{2} - 2 \lambda_{\text{min}} \left(\hat{ \boldsymbol{\Lambda}}  \boldsymbol{\Lambda}\right) \left\{ \left\| \tilde{\mathbf{K}}_{x} \right\|_{F}^{2} + \left\| \tilde{\mathbf{K}}_{r} \right\|_{F}^{2} + \left\| \tilde{ \boldsymbol{\Theta}} \right\|_{F}^{2} \right\} + 2 \lambda_{\text{max}}\left(\mathbf{P}\right) \left\| \mathbf{e} \right\| \overline{\xi} \\
		& + 2\left\| \tilde{\mathbf{K}}_{x} \right\|_{F} \left\| \left( \mathbf{B}^{\dagger} \tilde{\mathbf{A}} + \tilde{ \boldsymbol{\Lambda}} \mathbf{K}_{x}^{\top} \right)^{\top}  \boldsymbol{\Lambda} \right\|_{F} + 2 \left\| \tilde{\mathbf{K}}_{r} \right\|_{F} \left\| \mathbf{K}_{r} \tilde{ \boldsymbol{\Lambda}}  \boldsymbol{\Lambda} \right\|_{F} + 2 \left\| \tilde{ \boldsymbol{\Theta}} \right\|_{F} \left\|  \boldsymbol{\Theta} \tilde{ \boldsymbol{\Lambda}}  \boldsymbol{\Lambda} \right\|_{F} - \left\| \tilde{\mathbf{w}} \right\|^{2} + \overline{\xi}_{fb}^{2},
	\end{aligned}
\end{equation*}
Further using $2ab \leq \epsilon a^{2}+b^{2}/\epsilon$ for any $a,b \in \mathbb{R}$ and $\epsilon = \lambda_{\text{min}} \left(\hat{ \boldsymbol{\Lambda}}  \boldsymbol{\Lambda}\right)$.
\begin{equation}
	\begin{aligned}
		\dot{V} & \leq - \dfrac{1}{2} \lambda_{\text{min}}\left(\mathbf{Q}\right) \left\| \mathbf{e} \right\|^{2} - \lambda_{\text{min}} \left(\hat{ \boldsymbol{\Lambda}}  \boldsymbol{\Lambda}\right) \left\{ \left\| \tilde{\mathbf{K}}_{x} \right\|_{F}^{2} + \left\| \tilde{\mathbf{K}}_{r} \right\|_{F}^{2} + \left\| \tilde{ \boldsymbol{\Theta}} \right\|_{F}^{2} \right\} - \left\| \tilde{\mathbf{w}} \right\|^{2} + \underbrace{ \overline{\xi}_{fb}^{2}}_{\rho_{w}} \\
		& \quad + \underbrace{2 \overline{\xi}^{2} \dfrac{\lambda_{\text{max}}^{2}\left(\mathbf{P}\right)}{\lambda_{\text{min}}\left(\mathbf{Q}\right)}}_{\rho_{1}} + \underbrace{\dfrac{1}{\lambda_{\text{min}} \left(\hat{ \boldsymbol{\Lambda}}  \boldsymbol{\Lambda}\right)} \left\{ \left\| \left( \mathbf{B}^{\dagger} \tilde{\mathbf{A}} + \tilde{ \boldsymbol{\Lambda}} \mathbf{K}_{x}^{\top} \right)^{\top}  \boldsymbol{\Lambda} \right\|_{F}^{2} + \left\| \mathbf{K}_{r} \tilde{ \boldsymbol{\Lambda}}  \boldsymbol{\Lambda} \right\|_{F}^{2} + \left\|  \boldsymbol{\Theta} \tilde{ \boldsymbol{\Lambda}}  \boldsymbol{\Lambda} \right\|_{F}^{2} \right\}}_{\rho_{3}}.
	\end{aligned}
	\label{eq:V_dot_step2}
\end{equation}
The \textit{class}-$\mathcal{K}$ functions defined in $\eqref{eq:gamma_functions_gradient}$ are modified as follows.
\begin{equation}
	\begin{aligned}
		\alpha_{2} \left(\left\| \boldsymbol{\chi}\right\|\right) &= \text{min} \left( \lambda_{\text{min}} \left( \mathbf{P} \right), \lambda_{\text{min}} \left( \boldsymbol{\Lambda}_{s} \boldsymbol{\Lambda} \right),1\right) \left\|\boldsymbol{\chi}\right\|^{2}, \\
		\beta_{2} \left(\left\| \boldsymbol{\chi}\right\|\right) &= \text{max} \left( \lambda_{\text{max}} \left( \mathbf{P} \right), \lambda_{\text{max}} \left( \boldsymbol{\Lambda}_{s} \boldsymbol{\Lambda} \right),1\right) \left\|\boldsymbol{\chi}\right\|^{2},\\
		\gamma_{2} \left(\left\| \boldsymbol{\chi}\right\|\right) &= \text{min} \left( 0.5 \lambda_{\text{min}} \left( \mathbf{Q} \right), \lambda_{\text{min}} \left( \hat{\boldsymbol{\Lambda}} \boldsymbol{\Lambda} \right),1 \right) \left\|\boldsymbol{\chi}\right\|^{2}.
	\end{aligned} 
	\label{eq:gamma_functions_proposed}
\end{equation}
The Lyapunov function $V$ in $\eqref{eq:V_combined}$ satisfies $\alpha_{2} \left( \left\| \boldsymbol{\chi} \right\|\right) \leq V \leq \beta_{2} \left( \left\| \boldsymbol{\chi} \right\|\right)$, and its derivative along the error trajectories $\eqref{eq:tracking_error}$ and $\eqref{eq:controller_gain_estimation_error_dynamics}$ satisfies the following.
\begin{equation}
	\dot{V} \leq - \gamma_{2} \left(\left\| \boldsymbol{\chi}\right\|\right) + \rho_{1} + \rho_{3} + \rho_{w},
	\label{eq:V_dot_proposed}
\end{equation}
where $\rho_{1}$ and $\rho_{w}$ depend on nonparametric uncertainty bound, system constant parameters, and designed gains, while $\rho_{3}$ depends on the ideal controller gains $\mathbf{K}_{x}$ and $\mathbf{K}_{r}$, the ideal system parameter $\boldsymbol{\Theta}$, and the estimation errors $\tilde{\mathbf{A}}$ and $\tilde{\boldsymbol{\Lambda}}$.
These estimation errors evolve under the adaptation law $\eqref{eq:adaptation_law_system_parameter}$. Therefore, using Theorem $\ref{theorem:parameter_estimation_bounds}$ and Proposition $\ref{prop:parameter_estimation_error_bound}$, $\rho_{3}$ can be expressed in terms of the ultimately bounded $\tilde{\mathbf{w}}$, constant system parameters, and controller gains.
\begin{equation*}
	\rho_{3} \leq \dfrac{1}{\lambda_{\text{min}} \left(\hat{\boldsymbol{\Lambda}} \boldsymbol{\Lambda}\right)} \left\{ \left( 1 + \left\| \mathbf{K}_{x} \right\|_{F} \right)^{2} + \left\| \mathbf{K}_{r} \right\|_{F}^{2} + \left\| \boldsymbol{\Theta} \right\|_{F}^{2} \right\} \left\| \boldsymbol{\Lambda} \right\|_{F}^{2} \left\| \tilde{\mathbf{w}} \right\|^{2}.
\end{equation*}
The derivative $\dot{V}$ from $\eqref{eq:V_dot_proposed}$ can now be written compactly as follows.
\begin{equation*}
	\dot{V} \leq -k V + \rho, \quad \rho=\rho_{1} + \rho_{3} + \rho_{w}, \quad k = \dfrac{\text{min} \left( 0.5 \lambda_{\text{min}} \left( \mathbf{Q} \right), \lambda_{\text{min}} \left( \hat{\boldsymbol{\Lambda}} \boldsymbol{\Lambda} \right), 1 \right)}{\text{max} \left( \lambda_{\text{max}} \left( \mathbf{P} \right), \lambda_{\text{max}} \left( \boldsymbol{\Lambda}_{s} \boldsymbol{\Lambda} \right),1\right)},
\end{equation*}
whose solution, using \textit{Lemma A.3.2, \cite{farrell2006adaptive}}, is
\begin{equation*}
	V\left(t\right) \leq \left(V\left(t_{q}\right) - \dfrac{\rho}{k}\right)e^{-k \left(t-t_{q}\right)} + \dfrac{\rho}{k}, \quad \forall \, t>t_{q}.
\end{equation*}
Expressing the solution in terms of the norm of the combined error vector.
\begin{equation*}
	\left\| \boldsymbol{\chi}\left(t\right) \right\|^{2} \leq \dfrac{1}{\text{min} \left( \lambda_{\text{min}} \left( \mathbf{P} \right), \lambda_{\text{min}} \left( \boldsymbol{\Lambda}_{s} \boldsymbol{\Lambda} \right),1\right)} \left(V\left(t_{q}\right) - \dfrac{\rho}{k}\right)e^{-k \left(t-t_{q}\right)} + \dfrac{1}{\text{min} \left( \lambda_{\text{min}} \left( \mathbf{P} \right), \lambda_{\text{min}} \left( \boldsymbol{\Lambda}_{s} \boldsymbol{\Lambda} \right),1\right)} \dfrac{\rho}{k}, \quad \forall \, t>t_{q},
\end{equation*}
and substituting $V\left(t_{q}\right) \leq \beta_{2}\left(\left\| \boldsymbol{\chi}\left(t_{q}\right) \right\|\right)$, we obtain the following.
\begin{equation*}
	\left\| \boldsymbol{\chi}\left(t\right) \right\|^{2} \leq \kappa_{1} \left\| \boldsymbol{\chi}\left(t_{q}\right) \right\|^{2} e^{-k \left(t-t_{q}\right)} + \kappa_{2} \left(1-e^{-k\left(t-t_{q}\right)}\right) \rho,
\end{equation*}
where
\begin{equation*}
	\kappa_{1} = \dfrac{\text{max} \left( \lambda_{\text{max}} \left( \mathbf{P} \right), \lambda_{\text{max}} \left( \boldsymbol{\Lambda}_{s} \boldsymbol{\Lambda} \right),1 \right)}{\text{min} \left( \lambda_{\text{min}} \left( \mathbf{P} \right), \lambda_{\text{min}} \left( \boldsymbol{\Lambda}_{s} \boldsymbol{\Lambda} \right), 1 \right)}, \quad \kappa_{2} = \dfrac{1}{\text{min} \left( 0.5 \lambda_{\text{min}} \left( \mathbf{Q} \right), \lambda_{\text{min}} \left( \hat{\boldsymbol{\Lambda}} \boldsymbol{\Lambda} \right), 1 \right)} \kappa_{1}.
\end{equation*}
The exponential decay rate $k$ does not explicitly depend on the excitation level, though the time $t_q$ at which this rate becomes active does. This completes the proof.
\end{proof}

\begin{remark}
The exponential decay rate is given by
\begin{equation*}
	k = \dfrac{\text{min} \left( 0.5 \lambda_{\text{min}} \left( \mathbf{Q} \right),k_{\Lambda}, 1 \right)}{\text{max} \left( \lambda_{\text{max}} \left( \mathbf{P} \right), \lambda_{\text{max}} \left( \boldsymbol{\Lambda}_{s} \boldsymbol{\Lambda} \right),1\right)},
\end{equation*}
where $k_{\Lambda}$ is a positive lower bound on $\lambda_{\text{min}} \left( \hat{\boldsymbol{\Lambda}} \boldsymbol{\Lambda} \right)$ established via Assumption $\ref{assump:Lambda}$ and the conditioning in $\eqref{eq:condition_Lambda}$, ensuring $k$ is a positive constant independent of the time-varying estimate $\hat{ \boldsymbol{\Lambda}}$. The ultimate bound is
\begin{equation*}
	\rho_{ub} =  \kappa_{2} \left(1-e^{-k\left(t-t_{q}\right)}\right) \left( \rho_{1} + \rho_{3} + \rho_{w} \right)
\end{equation*}
where $\rho_{1}$ depends on the nonparametric uncertainty bound $\overline{\xi}$ and the Lyapunov matrices, $\rho_{w}$ depends on the transformed uncertainty bound $\overline{\xi}_{fb}$ from Theorem $\ref{theorem:algorithm1_lip}$, and $\rho_{3}$ depends on the ideal controller gains, system parameters, and the ultimately bounded estimation error $\left\| \tilde{\mathbf{w}} \right\|$. In the absence of nonparametric uncertainty, $\overline{\xi} = 0$ implies $\rho_{1} = \rho_{w} = 0$, and $\rho_{3} \to 0$ as $\left\| \tilde{\mathbf{w}} \right\| \to 0$, so that $\rho_{ub}\to  0$.
\end{remark}

\begin{remark}
The combined adaptation law established in Theorem $\ref{theorem:parameter_estimation_bounds}$ operates on a single orthogonal dataset constructed by Algorithm $\ref{alg:identification}$. While this guarantees exponential decay to the ultimate bound $\rho_{ub}$, the bound itself depends on the nonparametric uncertainty through $\rho_{w} = \overline{\xi}_{fb}^{2}$. Updating the parameter estimates using a single dataset is therefore not preferable when nonparametric uncertainties are significant, since the uncertainty realization in that dataset directly determines the ultimate bound. To address this, batch and recursive extensions of Algorithm $\ref{alg:identification}$ that systematically incorporate multiple datasets into the parameter update law can be utilized. These extensions should retain the linear-in-parameter structure established in Theorem $\ref{theorem:algorithm1_lip}$ and allow the effect of nonparametric uncertainties to be averaged across datasets, reducing their impact on the ultimate bound.
\end{remark}

\section{Numerical Simulations}
\label{section:Simulation_Results}

The proposed methodology is validated on three examples. Example 1 addresses parameter identification from a linear regression model and comprises two studies differing in the structure of the regressor: in the first, the regressor is a function of time alone, while in the second, it incorporates state variables governed by a known differential equation. These studies compare the proposed method against concurrent learning (CL), MRE, and DREM in terms of parameter convergence under the finite excitation condition. Example 2 applies the proposed combined adaptive controller to the longitudinal dynamics of an aircraft \cite{lavretsky2009combined} 

\subsection{Parameter Identification}

\textbf{Study 1.} Consider the following linear regression model \cite{ortega2020modified}.
\begin{equation*}
	y=\boldsymbol{\phi}^{\top}\mathbf{w}, \quad \boldsymbol{\phi} = \begin{bmatrix} 1 & \dfrac{\sin(t)+\cos(t)}{\sqrt{1+t}}-\dfrac{\sin(t)}{2\left(1+t\right)^{3/2}} \end{bmatrix}^{\top}, \quad \mathbf{w}=\begin{bmatrix} 1 & 2 \end{bmatrix}^{\top},
\end{equation*}
where $y$ is measurable and $\mathbf{w}$ is an unknown constant parameter vector. The regressor vector $\boldsymbol{\phi}$ satisfies the condition required by the dynamic regressor extension and mixing method for parameter convergence [see Proposition 2 and Remark 7 \cite{aranovskiy2016parameters}]. The parameter estimate is updated using $\eqref{eq:adaptation_law_system_parameter}$, with the memory term constructed separately using concurrent learning \cite{chowdhary2013concurrent}, MRE \cite{roy2017uges}, DREM \cite{aranovskiy2016parameters}, and the proposed method, each with learning rate $\gamma_l > 0$. Initial estimates are set to $\hat{\mathbf{w}}(0) = \begin{bmatrix} 0 & 0 \end{bmatrix}^{\top}$. Two values of the learning rate are studied. At $\gamma_l = 1$, Figure \ref{fig:pee_norm_gl_1} shows that DREM fails to drive the estimation error norm to zero within the simulation horizon, a consequence of the small-determinant scaling issue noted in Section \ref{section:Introduction}. The learning rate is therefore increased to $\gamma_l = 10$ to provide a fair comparison at higher excitation. At this gain, Figure \ref{fig:pee_norm_gl_10} shows that the proposed method achieves the fastest exponential convergence among all methods. The time evolution of individual parameter estimates in both figures confirms that the proposed method converges to the true values $\mathbf{w} = [1, 2]^\top$ more rapidly and without overshoot.

\begin{figure}[h]
	\includegraphics[width=0.48\textwidth]{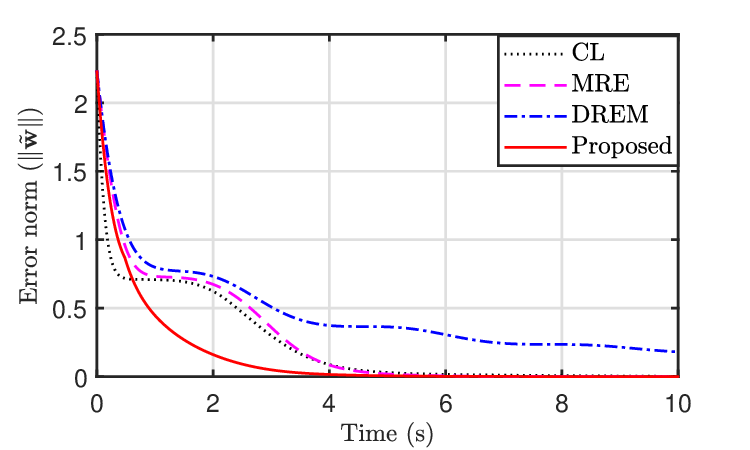}
	\includegraphics[width=0.48\textwidth]{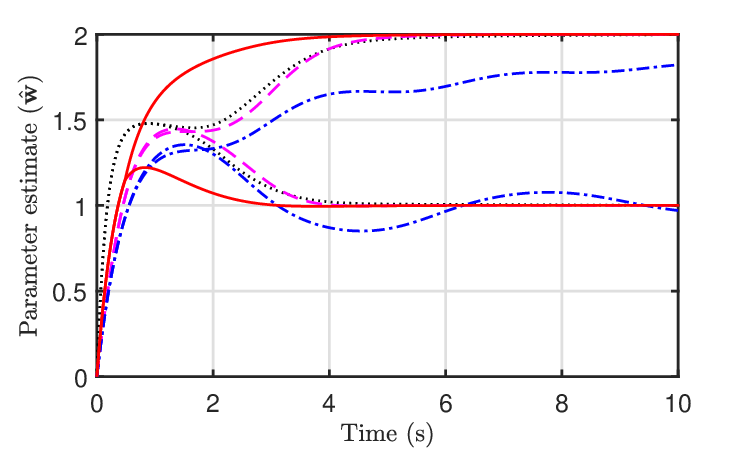}
	\caption{Comparison of the norm of the parameter estimation error and individual parameter estimate of the proposed method with different methods: concurrent learning \cite{chowdhary2013concurrent}, MRE \cite{roy2017uges}, DREM \cite{aranovskiy2016parameters} with $\gamma_{l}=1$.}
	\label{fig:pee_norm_gl_1}
\end{figure}

\begin{figure}[h]
	\includegraphics[width=0.48\textwidth]{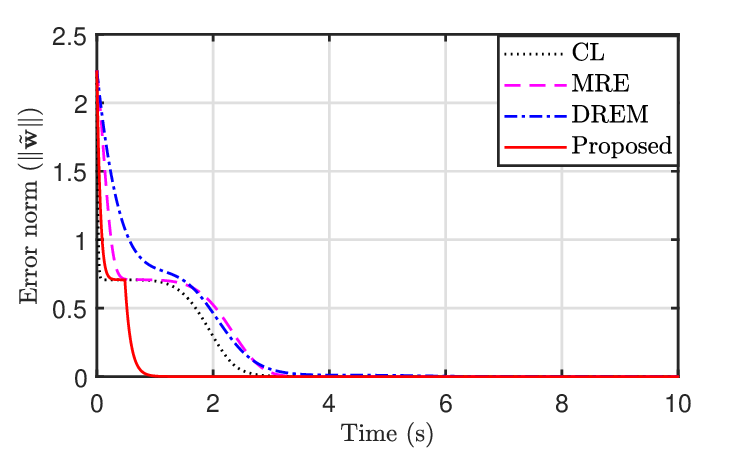}
	\includegraphics[width=0.48\textwidth]{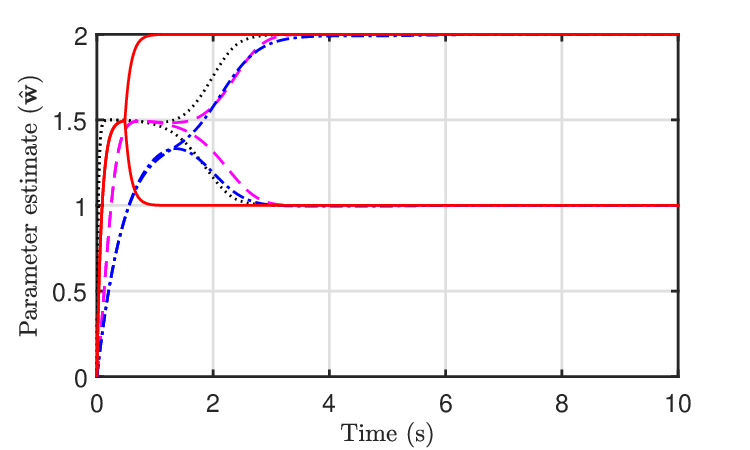}
	\caption{Comparison of the norm of the parameter estimation error and individual parameter estimate of the proposed method with different methods: concurrent learning \cite{chowdhary2013concurrent}, MRE \cite{roy2017uges}, DREM \cite{aranovskiy2016parameters} with $\gamma_{l}=10$.}
	\label{fig:pee_norm_gl_10}
\end{figure}

\textbf{Study 2.} The linear regression model is extended to incorporate a state-dependent regressor, providing a more practically relevant excitation scenario. The modified model is as follows.
\begin{equation*}
	y_{z}=\boldsymbol{\phi}_{z}^{\top} \mathbf{w}_{z}, \quad \boldsymbol{\phi}_{z} = \begin{bmatrix} \mathbf{z}^{\top} & \boldsymbol{\phi}^{\top}	\end{bmatrix}^{\top}, \quad y_{z}=y+\dot{z}_{2}, \quad \mathbf{w}_{z} = \begin{bmatrix} -1 & -1.4 & \mathbf{w}^{\top} \end{bmatrix}^{\top},
\end{equation*}
where $\mathbf{w}_{z}$ is an unknown parameter vector and $\mathbf{z}$ is the state vector governed by the following dynamics.
\begin{equation*}
	\dot{\mathbf{z}} = \begin{bmatrix} 0 & 1 \\ -1 & -1.4 \end{bmatrix} \mathbf{z} + \begin{bmatrix} 0 \\ 1 \end{bmatrix}, \quad \mathbf{z} = \begin{bmatrix} z_{1} & z_{2}	\end{bmatrix}^{\top} \in \mathbb{R}^{2}.
\end{equation*}
For this study, it is assumed that $\dot{z}_{2}$ is measurable, and consequently, $y_{z}$ is also measurable. Using $\gamma_{l}=1$, Figure \ref{fig:pee_norm_gl_1_st_2} shows that only the proposed method drives the estimation error norm exponentially to the origin. The remaining methods fail to achieve convergence at this gain.

To provide a fair comparison, tuned learning rates are applied to each method: 
$\gamma_{l}=\dfrac{200}{\lambda_{\text{max}}\left(\boldsymbol{\Psi}_{m}\right)}$ for CL, $\gamma_{l}=70$ for MRE, and $\gamma_{l}=10^{7}$ for DREM. For CL, the minimum and maximum eigenvalues of the coefficient matrix are of order $10^{-15}$ and $8$, respectively. The severe ill-conditioning requires an excessively large learning rate to compensate for the small minimum eigenvalue, which induces stiffness in the parameter update differential equation and renders CL numerically impractical for this example.

Except for the concurrent learning method, all methods achieve convergence. The parameter estimation error norm for all methods with tuned learning rates is shown in Figure $\ref{fig:norm_gl_req}$. It is to be noted that the learning rate is kept same $\gamma_{l} = 1$ for the proposed method. MRE and DREM achieve convergence at their respective tuned rates; results are shown in Figure \ref{fig:norm_gl_req}. The proposed method achieves convergence with $\gamma_{l} = 1$, unchanged from the baseline comparison. For DREM, the large learning rate is admissible because the coefficient matrix is an identity matrix scaled by $(\det)^2$, with $(\det)^2 \approx 10^{-7}$ here, so $\gamma_{l}=10^{7}$ achieves adequate convergence. However, this determinant depends sensitively on the filter frequency, so time-varying gains would be required to maintain performance across operating conditions.

The results of this study highlight a fundamental distinction between the proposed method and existing approaches. Concurrent learning, MRE, and DREM all exhibit sensitivity to the level of regressor excitation: their convergence rates depend directly on the eigenvalues of their respective coefficient matrices, which in turn depend on the excitation level of the regressor. Consequently, achieving satisfactory parameter convergence requires careful tuning of the learning rate for each method, and a gain that works well under one excitation level may be inadequate or destabilizing under another, as demonstrated by the large learning rates required in this study. The proposed method, by contrast, yields an identity coefficient matrix in the parameter estimation error dynamics under the finite excitation condition, making the convergence rate independent of the regressor excitation level and determined solely by the learning rate $\gamma_l$. This eliminates the need for excitation-aware gain tuning and ensures consistent convergence performance across varying excitation conditions, as confirmed by the results at $\gamma_l = 1$ throughout this study.

\begin{figure}[h]
	\centering
	\includegraphics[width=0.48\textwidth]{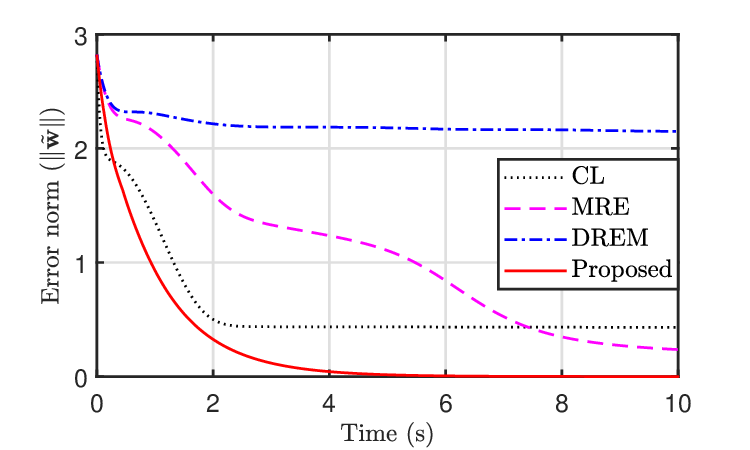}
	\caption{Time evolution of the norm of the parameter estimation error using different adaptation methods with $\gamma_{l}=1$.}
	\label{fig:pee_norm_gl_1_st_2}
\end{figure}

\begin{figure}[h]
	\centering
	\includegraphics[width=0.48\textwidth]{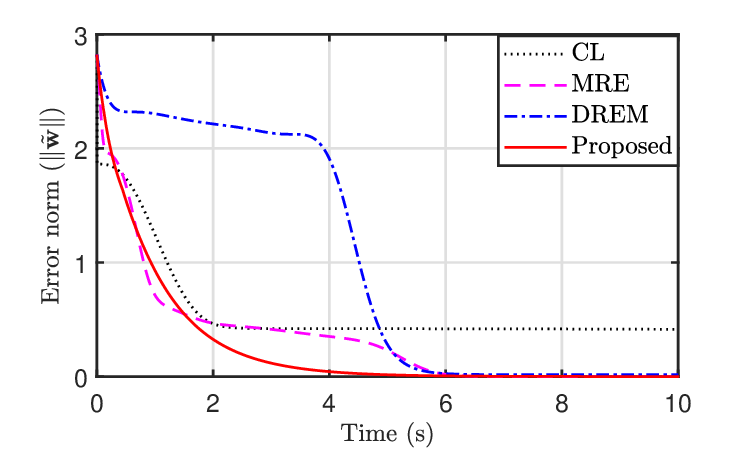}
	\caption{Time evolution of the norm of the parameter estimation error using different adaptation methods with tuned learning rates while $\gamma_{l}=1$ for the proposed method.}
	\label{fig:norm_gl_req}
\end{figure}



\subsection{Aircraft Longitudinal Control}
Consider the following longitudinal dynamics of an aircraft as presented in \cite{lavretsky2009combined}.
\begin{equation*}
	\begin{bmatrix} \dot{\alpha} \\ \dot{q}	\end{bmatrix} = \begin{bmatrix} -1.0189 & 0.9051 \\ 0.8223 & -1.0774	\end{bmatrix} \begin{bmatrix} \alpha \\ q \end{bmatrix} + \begin{bmatrix} -0.0022 \\ -0.1756 \end{bmatrix} \delta_{e},
\end{equation*}
where $\alpha$ is the angle of attack (in radians), $q$ is the pitch rate (in radians/second), and $\delta_{e}$ is the elevator deflection (in degrees), which serves as the control input to the system. 

Three kinds of matched uncertainties are introduced to the longitudinal dynamics of the aircraft system \cite{lavretsky2009combined}. A linear-in-state uncertainty $\boldsymbol{\phi}_{1}\left(\alpha,q\right)$, a constant control effectiveness uncertainty $\Lambda$, and a nonlinear-in-state uncertainty $\boldsymbol{\phi}_{2}\left(\alpha,q\right)$. The modified dynamics is as follows.
\begin{equation*}
	\begin{bmatrix} \dot{\alpha} \\ \dot{q}	\end{bmatrix} = \begin{bmatrix} -1.0189 & 0.9051 \\ 0.8223 & -1.0774	\end{bmatrix} \begin{bmatrix} \alpha \\ q \end{bmatrix} + \begin{bmatrix} -0.0022 \\ -0.1756 \end{bmatrix} \Lambda \left( \delta_{e} + \boldsymbol{\phi}_{1}^{\top} \boldsymbol{\theta}_{1} + \boldsymbol{\phi}_{2}^{\top} \boldsymbol{\theta}_{2}  \right),
\end{equation*} 
where 
\begin{equation*}
	\boldsymbol{\phi}_{1} = \begin{bmatrix} \alpha \\ q \end{bmatrix}, \quad \boldsymbol{\phi}_{2} = \begin{bmatrix} e^{-\frac{\left(\alpha-6^{\circ}\right)}{2\sigma^{2}}} & e^{-\frac{\left(\alpha-4^{\circ}\right)}{2\sigma^{2}}} & e^{-\frac{\left(\alpha-2^{\circ}\right)}{2\sigma^{2}}} & e^{-\frac{\alpha}{2\sigma^{2}}} & e^{-\frac{\left(\alpha+2^{\circ}\right)}{2\sigma^{2}}} & e^{-\frac{\left(\alpha+4^{\circ}\right)}{2\sigma^{2}}} & e^{-\frac{\left(\alpha+6^{\circ}\right)}{2\sigma^{2}}} \end{bmatrix}^{\top},
\end{equation*}
with $\sigma=0.0233$, and the unknown constants are
\begin{equation*}
	\Lambda=0.5, \quad \boldsymbol{\theta}_{1} = \begin{bmatrix} -4.6836 & -9.8197 \end{bmatrix}^{\top}, \quad \boldsymbol{\theta}_{2} = \begin{bmatrix} 0.1 & 0.1 & 0.1 & 0.1 & 0.1 & 0.1 & 0.1 \end{bmatrix}^{\top}.
\end{equation*}
To facilitate command angle of attack tracking, the aircraft dynamics is modified to include the integral of tracking error as follows.
\begin{equation*}
	\begin{bmatrix}	\dot{e}_{I} \\ \dot{\alpha} \\ \dot{q} \end{bmatrix} = \begin{bmatrix} 0 & 1 & 0 \\ 0 & -1.0189 & 0.9051 \\ 0 & 0.8223 & -1.0774 \end{bmatrix} \begin{bmatrix} e_{I} \\ \alpha \\ q \end{bmatrix} + \begin{bmatrix} 0 \\ -0.0022 \\ -0.1756	\end{bmatrix} \Lambda \left( \delta_{e} + \boldsymbol{\phi}_{1}^{\top} \boldsymbol{\theta}_{1} + \boldsymbol{\phi}_{2}^{\top} \boldsymbol{\theta}_{2}  \right) + \begin{bmatrix} -1 \\ 0 \\ 0	\end{bmatrix} \alpha_{\text{cmd}},
\end{equation*}
where $\dot{e}_{I}=\alpha-\alpha_{\text{cmd}}$, and $\alpha_{\text{cmd}}$ is a bounded, time-varying commanded angle of attack. The corresponding reference model $\eqref{eq:reference_system}$ is as follows.
\begin{equation*}
	\begin{bmatrix}	\dot{e}_{Ir} \\ \dot{\alpha}_{r} \\ \dot{q}_{r} \end{bmatrix} = \begin{bmatrix} 0 & 1 & 0 \\ -0.0220 & -1.0428 & 0.8918 \\ -1.7560 & -1.0880 & -2.1413	\end{bmatrix} \begin{bmatrix} e_{Ir} \\ \alpha_{r} \\ q_{r} \end{bmatrix} + \begin{bmatrix} -1 \\ 0 \\ 0	\end{bmatrix} \alpha_{\text{cmd}}.
\end{equation*}
The baseline control law, $\delta_{e}=\mathbf{k}_{x}^{\top} \begin{bmatrix} e_{I} & \alpha & q \end{bmatrix}^{\top}$, satisfies the matching condition $\eqref{eq:matching_condition}$ with the nominal gain  $\mathbf{k}_{x}=\begin{bmatrix} 10 & 10.8786 & 6.0589 \end{bmatrix}^{\top}$ when $\Lambda=1$. To accommodate uncertainties, the control law is augmented.
\begin{equation*}
	\delta_{e} = \mathbf{k}_{x}^{\top} \begin{bmatrix} e_{I} & \alpha & q \end{bmatrix}^{\top} - \hat{\mathbf{k}}^{\top} \begin{bmatrix} e_{I} & \alpha & q \end{bmatrix}^{\top} - \boldsymbol{\phi}_{1}^{\top} \hat{\boldsymbol{\theta}}_{1} + \boldsymbol{\phi}_{2}^{\top} \hat{\boldsymbol{\theta}}_{2},
\end{equation*}
where $\hat{\mathbf{k}}$ compensates for deviations in the nominal gain due to the uncertainty in the control effectiveness constant $\Lambda$. The estimates  $\hat{\mathbf{k}}$, $\hat{\boldsymbol{\theta}}_{1}$, and $\hat{\boldsymbol{\theta}}_{2}$ are updated using the adaptation laws in $\eqref{eq:controller_gain_update_law_proposed}$. All controller and adaptation gains are selected as given in \cite{lavretsky2009combined}, except for the memory adaptation gain, which is set to unity in  $\eqref{eq:adaptation_law_system_parameter}$, $\underline{\lambda}=0.25$. The algorithm parameters are $\delta_{1}=0.25$ and $\delta_{2}=0.005$. The time evolution of the angle of attack and the required elevator deflection are shown in Figure $\ref{fig:aoa_delta_e}$. Once the finite excitation condition is verified online, the proposed method achieves precise angle of attack tracking with reduced oscillations in the elevator deflection, attributed to improved parameter estimation under the combined adaptation law. Unlike in \cite{lavretsky2009combined}, where parameter estimates vary with each command transition, the proposed memory-based approach maintains stable estimates across such transitions, resulting in smaller transients in the elevator deflection profile in response to changes in the reference signal. The norms of the parameter estimation errors $\left\| \tilde{\Lambda} \right\|$, $\left\| \tilde{\boldsymbol{\theta}}_{1} \right\|$, $\left\| \tilde{\boldsymbol{\theta}}_{2} \right\|$, and the controller gain estimation error $\left\| \tilde{\mathbf{k}} \right\|$ are shown in Figures  $\ref{fig:parameter_norm}$ and $\ref{fig:gain_error_norm}$, respectively. The tracking error norm shown in Figure $\ref{fig:tracking_error_norm}$ confirms superior closed-loop performance with the proposed combined adaptive controller. All estimation errors converge to the origin under the proposed combined adaptive controller, confirming that the finite excitation condition is sufficient for exponential parameter convergence in the absence of nonparametric uncertainties.


\begin{figure}[h]
	\centering
	\includegraphics[width=0.48\textwidth]{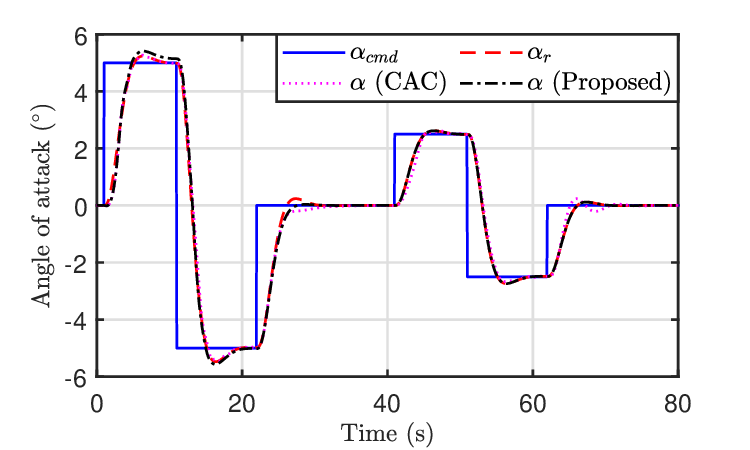}
	\includegraphics[width=0.48\textwidth]{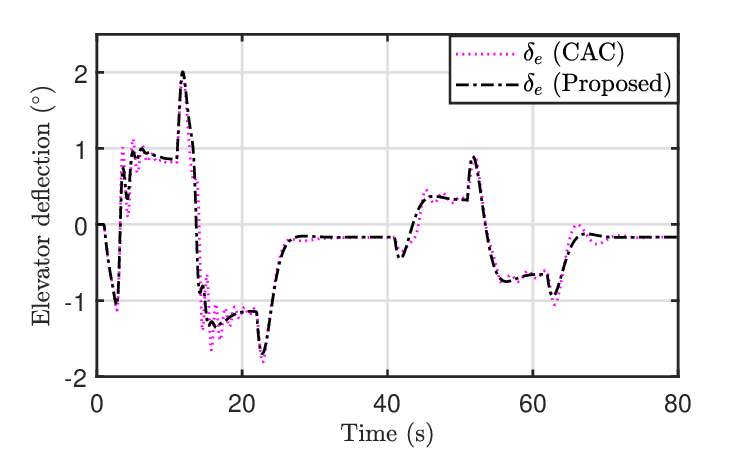}
	\caption{Time evolution of the angle of attack and the elevator deflection using the combined adaptive controller \cite{lavretsky2009combined} and the proposed method.}
	\label{fig:aoa_delta_e}
\end{figure}

\begin{figure}[h]
	\centering
	\includegraphics[width=0.48\textwidth]{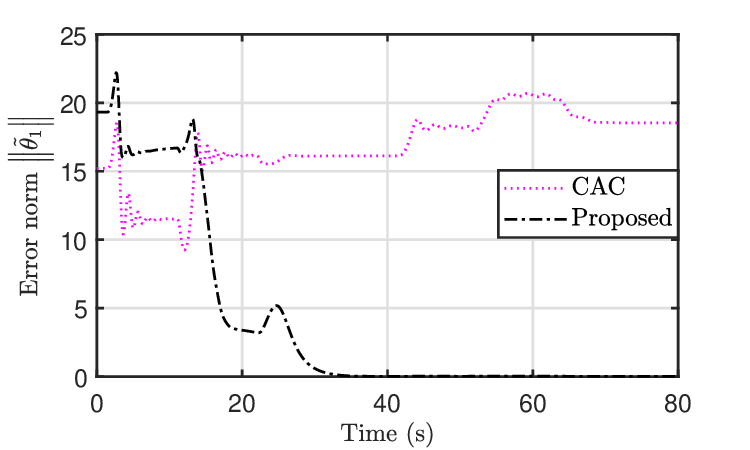}
	\includegraphics[width=0.48\textwidth]{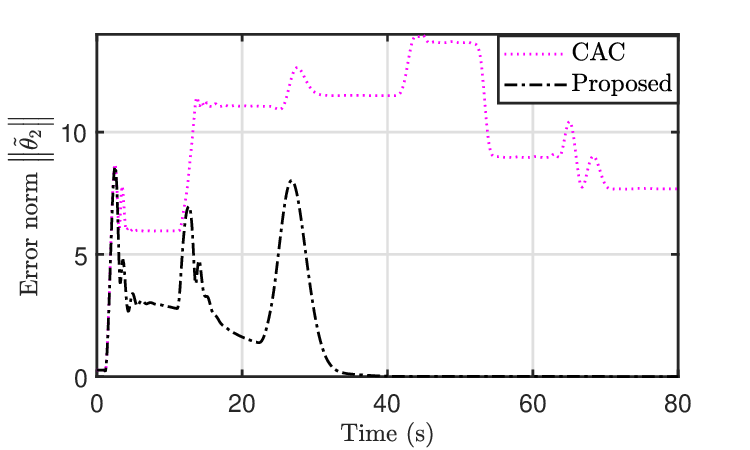}
	\caption{Time evolution of the parameter estimation errors norm using the combined adaptive controller \cite{lavretsky2009combined} and the proposed method.}
	\label{fig:gain_error_norm}
\end{figure}

\begin{figure}[h]
	\centering
	\includegraphics[width=0.48\textwidth]{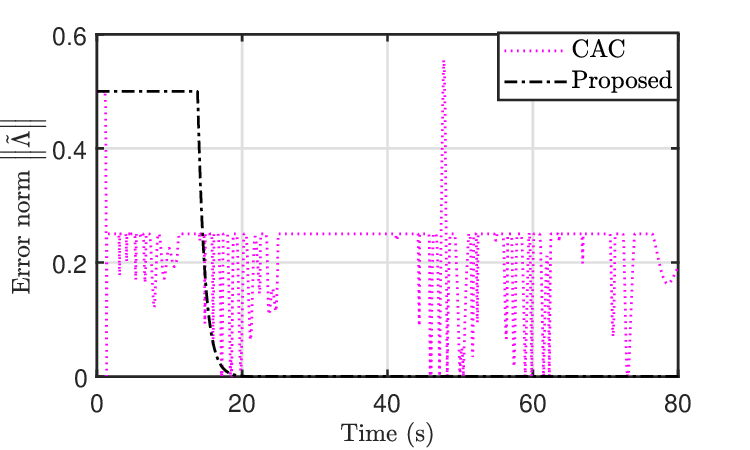}
	\includegraphics[width=0.48\textwidth]{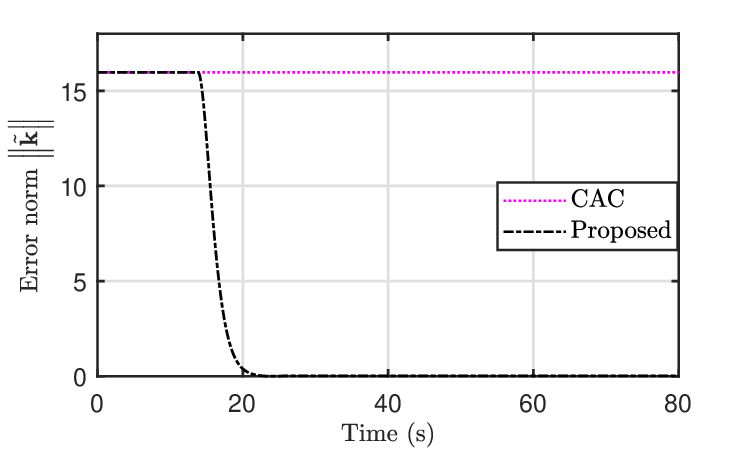}
	\caption{Time evolution of the parameter and controller gain estimation errors norm using the combined adaptive controller \cite{lavretsky2009combined} and the proposed method.}
	\label{fig:parameter_norm}
\end{figure}

\begin{figure}[h]
	\centering
	\includegraphics[width=0.48\textwidth]{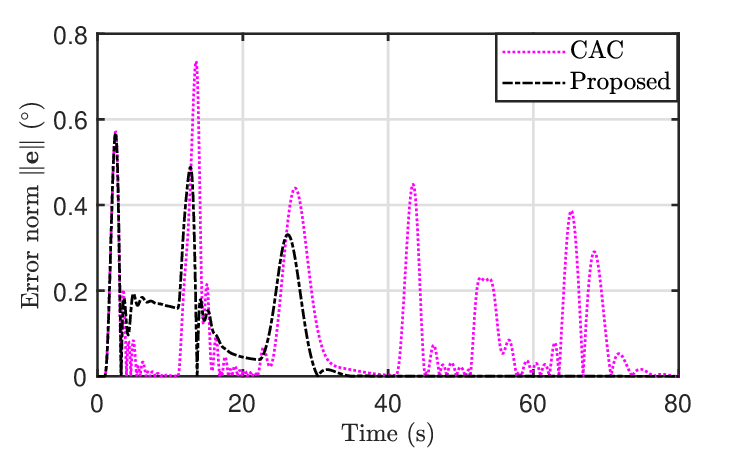}
	\caption{Time evolution of the tracking error norm using the combined adaptive controller \cite{lavretsky2009combined} and the proposed method.}
	\label{fig:tracking_error_norm}
\end{figure}

\section{Conclusion}
\label{section:conclusion}
This paper presented a robust combined adaptive control framework for a class of nonlinear multi-input multi-output systems with an unknown diagonal control effectiveness matrix, unknown system parameters, and bounded nonparametric uncertainties. A new algorithm was developed to construct a memory term via the Modified Gram-Schmidt orthogonalization, yielding an identity coefficient matrix in the parameter estimation error dynamics under the finite excitation condition. This structure eliminated the need for time-varying adaptation gains and enabled the derivation of an explicit ultimate bound on the parameter estimation error. A combined adaptation law was proposed in which controller gain updates are activated online upon verification of the finite excitation condition, guaranteeing exponential decay of the combined tracking and estimation errors to a neighborhood of the origin with a rate independent of regressor excitation. 

\section*{Appendix}
\subsection*{A - Proof of Proposition $\ref{proposition:direct_adaptive}$}
\begin{proof}
Consider the Lyapunov function candidate in terms of the tracking error and the estimation errors of the controller gains.
\begin{equation}
	V_{e} = \mathbf{e}^{\top} \mathbf{P} \mathbf{e} + \tr(\tilde{\mathbf{K}}_{x}^{\top} \tilde{\mathbf{K}}_{x}  \boldsymbol{\Lambda}_{s}  \boldsymbol{\Lambda}) + \tr( \tilde{\mathbf{K}}_{r}^{\top} \tilde{\mathbf{K}}_{r}  \boldsymbol{\Lambda}_{s}  \boldsymbol{\Lambda}) + \tr(\tilde{ \boldsymbol{\Theta}}^{\top} \tilde{ \boldsymbol{\Theta}}  \boldsymbol{\Lambda}_{s}  \boldsymbol{\Lambda}),
	\label{eq:V_gradient}
\end{equation}
where $ \boldsymbol{\Lambda}_{s}  \boldsymbol{\Lambda}$ is a positive definite diagonal matrix [\textit{Assumption} $\ref{assump:Lambda}$]. The time derivative of the Lyapunov function along the error trajectory $\eqref{eq:tracking_error}$ is
\begin{equation*}
	\begin{aligned}
		\dot{V}_{e} =& \mathbf{e}^{\top} \left( \mathbf{A}_{r}^{\top} \mathbf{P} + \mathbf{P} \mathbf{A}_{r} \right) \mathbf{e} + 2 \mathbf{e}^{\top} \mathbf{P} \mathbf{B}  \boldsymbol{\Lambda} \tilde{\mathbf{K}}_{x}^{\top} \mathbf{x} + 2 \mathbf{e}^{\top} \mathbf{P} \mathbf{B}  \boldsymbol{\Lambda} \tilde{\mathbf{K}}_{r}^{\top} \mathbf{r} - 2 \mathbf{e}^{\top} \mathbf{P} \mathbf{B}  \boldsymbol{\Lambda} \tilde{ \boldsymbol{\Theta}}^{\top} \boldsymbol{\phi}(\mathbf{x}) \\
		& + 2 \mathbf{e}^{\top} \mathbf{P} \boldsymbol{\xi} + 2 \tr(\tilde{\mathbf{K}}_{x}^{\top} \dot{\tilde{\mathbf{K}}}_{x}  \boldsymbol{\Lambda}_{s}  \boldsymbol{\Lambda}) + 2 \tr(\tilde{\mathbf{K}}_{r}^{\top} \dot{\tilde{\mathbf{K}}}_{r}  \boldsymbol{\Lambda}_{s}  \boldsymbol{\Lambda}) + 2 \tr(\tilde{ \boldsymbol{\Theta}}^{\top} \dot{\tilde{ \boldsymbol{\Theta}}}  \boldsymbol{\Lambda}_{s}  \boldsymbol{\Lambda}).
	\end{aligned}
\end{equation*}
The cross terms $2 \mathbf{e}^{\top} \mathbf{P} \mathbf{B}  \boldsymbol{\Lambda} \tilde{\mathbf{K}}_{x}^{\top} \mathbf{x}$ and $2 \tr(\tilde{\mathbf{K}}_{x}^{\top} \dot{\tilde{\mathbf{K}}}_{x}  \boldsymbol{\Lambda}_{s}  \boldsymbol{\Lambda})$ cancel upon substituting the update law $\eqref{eq:controller_gain_update_law_gradient}$ into the trace terms, using the identity $\text{tr} \left( \mathbf{m}_{1} \mathbf{m}_{2}^{\top} \right) = \mathbf{m}_{1}^{\top} \mathbf{m}_{2}$ (Preliminaries, identity 5) and the fact that $\boldsymbol{\Lambda}_{s}  \boldsymbol{\Lambda}$ is a diagonal positive definite matrix (Assumption 2); identical cancellations hold for the $\tilde{\mathbf{K}}_{r}$ and $\tilde{\boldsymbol{\Theta}}$ terms, yielding the simplified expression below.
\begin{equation*}
	\dot{V}_{e} = - \mathbf{e}^{\top} \mathbf{Q} \mathbf{e} + 2 \mathbf{e}^{\top} \mathbf{P} \boldsymbol{\xi} - 2 \tr(\tilde{\mathbf{K}}_{x}^{\top} \hat{\mathbf{K}}_{x}  \boldsymbol{\Lambda}_{s}  \boldsymbol{\Lambda}) - 2 \tr(\tilde{\mathbf{K}}_{r}^{\top} \hat{\mathbf{K}}_{r}  \boldsymbol{\Lambda}_{s}  \boldsymbol{\Lambda}) - 2 \tr(\tilde{ \boldsymbol{\Theta}}^{\top} \hat{ \boldsymbol{\Theta}}  \boldsymbol{\Lambda}_{s}  \boldsymbol{\Lambda}).
\end{equation*}
Applying the matrix inequalities discussed in Section $\ref{section:Preliminaries}$ and the scalar inequality $2ab \leq \epsilon a^{2} + b^{2}/\epsilon$ for any $a,b \in \mathbb{R}$ and $\epsilon = \lambda_{\text{min}} \left(  \boldsymbol{\Lambda}_{s}  \boldsymbol{\Lambda} \right)$, we further obtain
\begin{equation}
	\begin{aligned}
		\dot{V}_{e} \leq & - \dfrac{1}{2} \lambda_{\text{min}} \left( \mathbf{Q} \right) \left\| \mathbf{e} \right\|^{2} - \lambda_{\text{min}} \left(  \boldsymbol{\Lambda}_{s}  \boldsymbol{\Lambda} \right) \left( \left\| \tilde{\mathbf{K}}_{x} \right\|_{F}^{2} + \left\| \tilde{\mathbf{K}}_{r} \right\|_{F}^{2} + \left\| \tilde{ \boldsymbol{\Theta}} \right\|_{F}^{2} \right) \\
		& + \underbrace{2 \overline{\xi}^{2} \dfrac{\lambda_{\text{max}}^{2} \left( \mathbf{P} \right)}{\lambda_{\text{min}} \left( \mathbf{Q} \right) }}_{\rho_{1}} + \underbrace{\dfrac{ \lambda_{\text{max}}^{2} \left(  \boldsymbol{\Lambda}_{s}  \boldsymbol{\Lambda} \right)}{\lambda_{\text{min}} \left(  \boldsymbol{\Lambda}_{s}  \boldsymbol{\Lambda} \right)} \left( \left\| \mathbf{K}_{x} \right\|_{F}^{2} + \left\| \mathbf{K}_{r} \right\|_{F}^{2} + \left\|  \boldsymbol{\Theta} \right\|_{F}^{2} \right)}_{\rho_{2}}.
	\end{aligned}
	\label{eq:Vdot_gradient}
\end{equation}
Let us define a few \textit{class}-$\mathcal{K}$ functions.
\begin{equation}
	\begin{cases}
		\alpha_{1} \left(\left\| \boldsymbol{\chi}_{e}\right\|\right) = \text{min} \left( \lambda_{\text{min}} \left( \mathbf{P} \right), \lambda_{\text{min}} \left(  \boldsymbol{\Lambda}_{s}  \boldsymbol{\Lambda} \right)\right) \left\|\boldsymbol{\chi}_{e}\right\|^{2}, & \\
		\beta_{1} \left(\left\| \boldsymbol{\chi}_{e}\right\|\right) = \text{max} \left( \lambda_{\text{max}} \left( \mathbf{P} \right), \lambda_{\text{max}} \left(  \boldsymbol{\Lambda}_{s}  \boldsymbol{\Lambda} \right)\right) \left\|\boldsymbol{\chi}_{e}\right\|^{2}, & \\
		\gamma_{1} \left(\left\| \boldsymbol{\chi}_{e}\right\|\right) = \text{min} \left( 0.5 \lambda_{\text{min}} \left( \mathbf{Q} \right), \lambda_{\text{min}} \left(  \boldsymbol{\Lambda}_{s}  \boldsymbol{\Lambda} \right) \right) \left\|\boldsymbol{\chi}_{e}\right\|^{2},
	\end{cases} \boldsymbol{\chi}_{e} = \begin{bmatrix} \mathbf{e}^{\top} & vec^{\top}\left(\tilde{\mathbf{K}}_{x}\right) & vec^{\top}\left(\tilde{\mathbf{K}}_{r}\right) & vec^{\top}\left(\tilde{ \boldsymbol{\Theta}}\right) \end{bmatrix}^{\top},
	\label{eq:gamma_functions_gradient}
\end{equation}
where $vec\left(\cdot\right)$ denotes the column-wise vectorization of a matrix, and $\boldsymbol{\chi}_{e}$ represents the combined error vector. The Lyapunov function $\eqref{eq:V_gradient}$ satisfies $\alpha_{1} \left(\left\| \boldsymbol{\chi}_{e}\right\|\right) \leq V_{e} \leq \beta_{1} \left(\left\| \boldsymbol{\chi}_{e}\right\|\right)$, and using $\eqref{eq:Vdot_gradient}$ and $\eqref{eq:gamma_functions_gradient}$, we obtain
\begin{equation*}
	\dot{V}_{e} \leq -\gamma_{1} \left(\left\| \boldsymbol{\chi}_{e}\right\|\right) + \rho_{1} + \rho_{2}.
\end{equation*}
This implies that all closed-loop errors are uniformly ultimately bounded, with an ultimate bound of $\rho_{1}+\rho_{2}$, [\textit{Theorem 13.1}, \cite{khalil2002nonlinear}]. Here, $\rho_{1}$ depends on the upper-bound of the nonparametric uncertainty $\overline{\xi}$, while $\rho_{2}$ depends on the ideal values of the controller gains $\mathbf{K}_{x}$ and $\mathbf{K}_{r}$, and the system parameter $\boldsymbol{\Theta}$. In the absence of nonparametric uncertainty $\boldsymbol{\xi}$, $\rho_{1}$ becomes zero, and the ultimate bound shrinks to $\rho_{2}$.
\end{proof}

\subsection*{B - Proof of Proposition $\ref{prop:parameter_estimation_error_bound}$}
\begin{proof}
From $\eqref{eq:system1}$ and $\eqref{eq:system_lip}$, $\mathbf{w} = vec\left(\begin{bmatrix} \mathbf{B}^{\dagger} \mathbf{A} &  \boldsymbol{\Lambda} &  \boldsymbol{\Lambda}  \boldsymbol{\Theta}^{\top} \end{bmatrix}^{\top}\right)$, and $\hat{\mathbf{w}}$ and $\tilde{\mathbf{w}}$ are obtained as follows.
\begin{equation*}
	\hat{\mathbf{w}} = vec\left(\begin{bmatrix} \mathbf{B}^{\dagger} \hat{\mathbf{A}} & \hat{ \boldsymbol{\Lambda}} & \hat{ \boldsymbol{\Lambda}} \hat{ \boldsymbol{\Theta}}^{\top} \end{bmatrix}^{\top}\right) \text{ and } \tilde{\mathbf{w}} = vec\left(\begin{bmatrix} \mathbf{B}^{\dagger} \tilde{\mathbf{A}} & \tilde{ \boldsymbol{\Lambda}} & \hat{ \boldsymbol{\Lambda}} \tilde{ \boldsymbol{\Theta}}^{\top} + \tilde{ \boldsymbol{\Lambda}}  \boldsymbol{\Theta}^{\top} \end{bmatrix}^{\top}\right).
\end{equation*}
Simplifying, $\hat{ \boldsymbol{\Lambda}} \hat{ \boldsymbol{\Theta}}^{\top} -  \boldsymbol{\Lambda}  \boldsymbol{\Theta}^{\top} = \hat{ \boldsymbol{\Lambda}} \hat{ \boldsymbol{\Theta}}^{\top} - \left(\hat{ \boldsymbol{\Lambda}} - \tilde{ \boldsymbol{\Lambda}}\right)  \boldsymbol{\Theta}^{\top} = \hat{ \boldsymbol{\Lambda}} \tilde{ \boldsymbol{\Theta}}^{\top} + \tilde{ \boldsymbol{\Lambda}}  \boldsymbol{\Theta}^{\top}$. Using the property that the Frobenius norm equals the Euclidean norm of the vectorized matrix.
\begin{equation*}
	\left\| \begin{bmatrix} \mathbf{B}^{\dagger} \tilde{\mathbf{A}} & \tilde{ \boldsymbol{\Lambda}} & \hat{ \boldsymbol{\Lambda}} \tilde{ \boldsymbol{\Theta}}^{\top} + \tilde{ \boldsymbol{\Lambda}}  \boldsymbol{\Theta}^{\top} \end{bmatrix} \right\|_{F} = \left\| \tilde{\mathbf{w}} \right\|,
\end{equation*}
which can be further simplified as following using $\left\| \begin{bmatrix} \mathbf{Q}_{1} & \mathbf{Q}_{2} \end{bmatrix} \right\|_{F}^{2} = \left\| \mathbf{Q}_{1} \right\|_{F}^{2} + \left\| \mathbf{Q}_{2} \right\|_{F}^{2}$
\begin{equation*}
	\left\| \mathbf{B}^{\dagger} \tilde{\mathbf{A}} \right\|_{F}^{2} + \left\| \tilde{ \boldsymbol{\Lambda}} \right\|_{F}^{2} + \left\| \hat{ \boldsymbol{\Lambda}} \tilde{ \boldsymbol{\Theta}}^{\top} + \tilde{ \boldsymbol{\Lambda}}  \boldsymbol{\Theta}^{\top} \right\|_{F}^{2} = \left\| \tilde{\mathbf{w}} \right\|^{2}.
\end{equation*}
Therefore, the following inequalities directly follow from the above relation
\begin{equation*}
	\begin{aligned}
		\left\| \mathbf{B}^{\dagger} \tilde{\mathbf{A}} \right\|_{F}^{2} \leq \left\| \tilde{\mathbf{w}} \right\|^{2} & \implies \left\| \mathbf{B}^{\dagger} \tilde{\mathbf{A}} \right\|_{F} \leq \left\| \tilde{\mathbf{w}} \right\|, \\
		\left\| \tilde{\boldsymbol{\Lambda}} \right\|_{F}^{2} \leq \left\| \tilde{\mathbf{w}} \right\|^{2} & \implies \left\| \tilde{\boldsymbol{\Lambda}} \right\|_{F} \leq \left\| \tilde{\mathbf{w}} \right\|.
	\end{aligned}
\end{equation*}
This completes the proof.
\end{proof}

\subsection*{C - Proof of Proposition $\ref{prop:controller_gain_estimation_error_dynamics}$}
\begin{proof}
The third term in $\dot{\hat{\mathbf{K}}}_{x}$ is expanded by substituting the equivalent expression for $\mathbf{A}_{r}$ from the matching condition $\eqref{eq:matching_condition}$.
\begin{equation*}
	\left( \mathbf{B}^{\dagger} \mathbf{A}_{r} - \mathbf{B}^{\dagger} \hat{\mathbf{A}} \right)^{\top} - \hat{\mathbf{K}}_{x} \hat{\boldsymbol{\Lambda}} = \left( \mathbf{B}^{\dagger} \left( \mathbf{A} + \mathbf{B} \boldsymbol{\Lambda} \mathbf{K}_{x}^{\top} \right) - \mathbf{B}^{\dagger} \hat{\mathbf{A}} \right)^{\top} - \hat{\mathbf{K}}_{x} \hat{\boldsymbol{\Lambda}}.
\end{equation*}
Substituting $\mathbf{A} = \hat{\mathbf{A}} - \tilde{\mathbf{A}}$ and $ \boldsymbol{\Lambda} = \hat{ \boldsymbol{\Lambda}} - \tilde{ \boldsymbol{\Lambda}}$ leads to the following.
\begin{equation*}
	\left( \mathbf{B}^{\dagger} \mathbf{A}_{r} - \mathbf{B}^{\dagger} \hat{\mathbf{A}} \right)^{\top} - \hat{\mathbf{K}}_{x} \hat{ \boldsymbol{\Lambda}} = \left( \mathbf{B}^{\dagger} \left( \hat{\mathbf{A}} - \tilde{\mathbf{A}} + \mathbf{B} \hat{ \boldsymbol{\Lambda}} \mathbf{K}_{x}^{\top} - \mathbf{B} \tilde{ \boldsymbol{\Lambda}} \mathbf{K}_{x}^{\top} \right) - \mathbf{B}^{\dagger} \hat{\mathbf{A}} \right)^{\top} - \hat{\mathbf{K}}_{x} \hat{ \boldsymbol{\Lambda}}.
\end{equation*}
Simplifying and collecting terms involving parameter estimation errors.
\begin{equation*}
	\left( \mathbf{B}^{\dagger} \mathbf{A}_{r} - \mathbf{B}^{\dagger} \hat{\mathbf{A}} \right)^{\top} - \hat{\mathbf{K}}_{x} \hat{ \boldsymbol{\Lambda}} = - \tilde{\mathbf{K}}_{x} \hat{ \boldsymbol{\Lambda}} + \left( -\mathbf{B}^{\dagger} \tilde{\mathbf{A}} - \tilde{ \boldsymbol{\Lambda}} \mathbf{K}_{x}^{\top} \right)^{\top}.
\end{equation*}
Substituting this expression into $\dot{\hat{\mathbf{K}}}_{x}$ results in the parameter error dynamics $\eqref{eq:controller_gain_estimation_error_dynamics}$. Similarly, the third term in the $\dot{\hat{\mathbf{K}}}_{r}$ is simplified as follows.
\begin{equation*}
	\begin{aligned}
		\left( \mathbf{B}^{\dagger} \mathbf{B}_{r} \right)^{\top} - \hat{\mathbf{K}}_{r} \hat{ \boldsymbol{\Lambda}} &= \left( \hat{ \boldsymbol{\Lambda}} \mathbf{K}_{r}^{\top} - \tilde{ \boldsymbol{\Lambda}} \mathbf{K}_{r}^{\top} \right)^{\top} - \hat{\mathbf{K}}_{r} \hat{ \boldsymbol{\Lambda}}, \\
		&= - \tilde{\mathbf{K}}_{r} \hat{ \boldsymbol{\Lambda}} - \mathbf{K}_{r} \tilde{ \boldsymbol{\Lambda}}.
	\end{aligned}
\end{equation*}
Substituting this into $\dot{\hat{\mathbf{K}}}_{r}$ gives the stated dynamics for $\tilde{\mathbf{K}}_{r}$. Finally, since  $\hat{ \boldsymbol{\Lambda}}_{ \boldsymbol{\Theta}}$ is the estimate of $ \boldsymbol{\Lambda}  \boldsymbol{\Theta}^{\top}$, it can be written as $\left(\hat{ \boldsymbol{\Lambda}} - \tilde{ \boldsymbol{\Lambda}}\right)  \boldsymbol{\Theta}^{\top}$. Substituting this into the third term of $\dot{\hat{ \boldsymbol{\Theta}}}$, we obtain the following.
\begin{equation*}
	\begin{aligned}
		\hat{ \boldsymbol{\Lambda}}_{ \boldsymbol{\Theta}}^{\top} - \hat{ \boldsymbol{\Theta}} \hat{ \boldsymbol{\Lambda}} &=  \boldsymbol{\Theta} \hat{ \boldsymbol{\Lambda}} -  \boldsymbol{\Theta} \tilde{ \boldsymbol{\Lambda}} - \hat{ \boldsymbol{\Theta}} \hat{ \boldsymbol{\Lambda}}, \\
		&= - \tilde{ \boldsymbol{\Theta}} \hat{ \boldsymbol{\Lambda}} -  \boldsymbol{\Theta} \tilde{ \boldsymbol{\Lambda}},
	\end{aligned}
\end{equation*}
which leads to the desired dynamics for $\tilde{ \boldsymbol{\Theta}}$. This completes the proof.
\end{proof}

\bibliographystyle{refsort} 
\bibliography{ReferenceList}

\end{document}